\documentclass[12pt]{article}

\usepackage[margin=1in]{geometry}
\usepackage{amsmath,amsthm,amssymb}
\usepackage{latexsym,xspace}
\usepackage{caption}
\usepackage{tikz}
\usetikzlibrary{arrows.meta}
\usetikzlibrary{decorations.markings}
\usepackage{kbordermatrix}
\usepackage{float}
\usepackage{url}
\usepackage{xspace}
\usepackage{nicefrac}
\usepackage[inline]{enumitem}
\usepackage[mathscr]{euscript}

\setlist[itemize,1]{leftmargin=2cm,labelsep=1cm,itemsep=20pt,topsep=10pt}
\setlist[enumerate,1]{topsep=5pt,itemsep=0pt,label=(\alph*)}

\floatstyle{plain}
\restylefloat{figure}

\newtheorem{theorem}{Theorem}
\newtheorem{lemma}[theorem]{Lemma}
\newtheorem{corollary}[theorem]{Corollary}

\theoremstyle{remark}

\newtheorem{example}{Example}

\newcommand{\thsp}{\hspace*{0.5pt}}

\newcommand{\hm}{\hspace{-1ex}}          % negative quad
\newcommand{\hq}{\hspace{0.5ex}}         % half quad
\newcommand{\ol}{\overline}              % complement
\newcommand{\sm}{\setminus}              % setminus
\newcommand{\tw}{\operatorname{tw}}      % treewidth

\newcommand{\dist}{\operatorname{dist}}
\newcommand{\diam}[1]{\operatorname{diam}(#1)}

\newcommand{\IN}{\mathbb{N}}  % positive integers
\newcommand{\Ptime}{\textsf{P}\xspace}
\newcommand{\NP}{\textsf{NP}\xspace}
\newcommand{\coNP}{\textsf{co-NP}\xspace}
\newcommand{\FP}{\textsf{FP}\xspace}
\newcommand{\numP}{\textsf{\#P}\xspace}
\newcommand{\Pspace}{\textsf{PSPACE}\xspace}
\newcommand{\Lspace}{\textsf{L}\xspace}

\newcommand{\tmix}{\ensuremath{\tau_{\textrm{mix}}}\xspace}
\newcommand{\qua}{\mbox{quasi-}}

\tikzset{
  vertex/.style     =
  { circle,
    inner sep       = -.1mm,
    fill            = none,
    minimum size    = 1.0mm,
    draw
  },
  class/.style      =
  { rectangle,
    minimum size    = 6mm,
    rounded corners = 1mm,
    inner sep       = 3pt,
    fill            = white,
    draw
  }
}
\tikzset{every picture/.style={line width=1pt}}
\tikzset{empty/.style={rectangle,draw=none,fill=none}}
\tikzset{b/.style = {circle,draw,inner sep=1pt,fill=black},
           w/.style = {circle,draw,inner sep=1pt,fill=none},
     }

\newcommand{\LR}{\ensuremath{L{:}R}}

\renewcommand{\emptyset}{\varnothing}
\newcommand{\cB}{\mathcal{B}}
\newcommand{\cC}{\mathcal{C}}
\newcommand{\cF}{\mathcal{F}}
\newcommand{\cG}{\mathcal{G}}
\newcommand{\cM}{\mathcal{M}}

\newcommand{\cS}{\mathcal{S}}
\newcommand{\eg}{\textsl{e.g.}\xspace}

\newcommand{\Adj}{\mathcal{N}}
\renewcommand{\deg}[1]{\ensuremath{\textrm{deg}(#1)}}
\newcommand{\nbh}[1]{\ensuremath{\mathcal{N}(#1)}}

\newcommand{\G}[1]{\ensuremath{\mathcal{G}(#1)}}
\newcommand{\jj}{\langle j\rangle}

\newcommand{\EUW}{E_{\mathrm{UW}}}
\newcommand{\EYZ}{E_{\mathrm{YZ}}}
\newcommand{\EW}{E_{\mathrm{W}}}
\newcommand{\EY}{E_{\mathrm{Y}}}
\newcommand{\EWX}{E_{\mathrm{WX}}}
\newcommand{\EXY}{E_{\mathrm{XY}}}
\newcommand{\EX}{E_{\mathrm{X}}}

\newcommand{\IE}{\mathbb{E}}  % even holes
\newcommand{\IH}{\mathbb{H}}  % all holes
\newcommand{\IO}{\mathbb{O}}  % odd holes
\newcommand{\IS}{\mathbb{S}}  % complete suns

\newcommand{\ohf}{\textsc{OddHoleFree}}
\newcommand{\ehf}{\textsc{EvenHoleFree}}
\newcommand{\perf}{\textsc{Perfect}}
\newcommand{\switch}{\textsc{Switchable}}
\newcommand{\wkch}{\textsc{WeakChordal}}
\newcommand{\bip}{\textsc{Bipartite}}
\newcommand{\och}{\textsc{OddChordal}}
\newcommand{\chord}{\textsc{Chordal}}
\newcommand{\cbg}{\textsc{ChordalBipartite}}
\newcommand{\strch}{\textsc{StrongChordal}}
\newcommand{\scs}{\textsc{StrongChordalSplit}}
\newcommand{\splt}{\textsc{Split}}
\newcommand{\conv}{\textsc{Convex}}
\newcommand{\tree}{\textsc{Forest}}
\newcommand{\bic}{\textsc{Biconvex}}
\newcommand{\perm}{\textsc{Permutation}}
\newcommand{\intl}{\textsc{Interval}}
\newcommand{\mono}{\textsc{Monotone}}
\newcommand{\qmon}{\textsc{QMonotone}}
\newcommand{\chains}{\textsc{Chains}}

\newcommand{\chp}{\textsc{ChordalPermutation}}
\newcommand{\Efcb}{\textsc{E-Free}}
\newcommand{\qchs}{\textsc{QChains}}
\newcommand{\cogr}{\textsc{Cograph}}
\newcommand{\unii}{\textsc{UnitInterval}}
\newcommand{\chain}{\textsc{Chain}}
\newcommand{\thre}{\textsc{Threshold}}
\newcommand{\qcbs}{\textsc{QCompletes}}
\newcommand{\cchn}{\textsc{Cochain}}
\newcommand{\cb}{\textsc{CompleteBipartite}}
\newcommand{\compl}{\textsc{Complete}}
\newcommand{\tww}{\textsc{TreeWidth\,$w$}}

\newcommand{\noflaw}{\textsc{Flawless}}

\newcommand{\upharpoons}{\ensuremath{\upharpoonleft\hspace{-0.63ex}\upharpoonright}}
\newcommand{\downharpoons}{\ensuremath{\downharpoonleft\hspace{-0.63ex}\downharpoonright}}
\newcommand{\ergo}{$\Downarrow$}      % ergodic
\newcommand{\nerg}{$\Uparrow$}        % not ergodic
\newcommand{\hard}{$\upuparrows$}     % #P-complete
\newcommand{\easy}{$\downdownarrows$} % polynomial
\newcommand{\rmix}{$\downarrow$}      % rapid mixing
\newcommand{\smix}{$\uparrow$}        % slow mixinig
\newcommand{\stab}{$\downharpoons$}   % P-stable
\newcommand{\nstb}{$\upharpoons$}     % not P-stable

\begin{document}

\title{Counting perfect matchings and the switch chain\thanks{The work reported in this paper was partially supported by EPSRC research grant EP/M004953/1 ``Randomized algorithms for computer networks''.}}
\author{Martin Dyer\thanks{School of Computing, University of Leeds, Leeds LS2~9JT, UK. Email: \texttt{m.e.dyer@leeds.ac.uk}. }
\and Haiko M\"{u}ller\thanks{School of Computing, University of Leeds, Leeds LS2~9JT, UK. Email: \texttt{h.muller@leeds.ac.uk}. }}

\date{January 22, 2018}

\maketitle

\begin{abstract}
We examine the problem of exactly or approximately counting all perfect matchings in hereditary classes of nonbipartite graphs. In particular, we consider the switch Markov chain of Diaconis, Graham and Holmes.
We determine the largest hereditary class for which the chain is ergodic, and define a large new hereditary
class of graphs for which it is rapidly mixing. We go on to show that the chain has exponential mixing time for a slightly larger class. We also examine the question of ergodicity of the switch chain in a arbitrary graph. Finally, we give exact counting algorithms for three classes.
\end{abstract}

\renewcommand{\arraystretch}{1.1}

\section{Introduction}\label{sec:intro}

In~\cite{DyJeMu17}, we examined (with Jerrum) the problem of counting all perfect matchings in some particular classes of bipartite graphs, inspired by a paper of Diaconis, Graham and Holmes~\cite{DiGrHo01} which gave applications to Statistics. That is, we considered the problem of evaluating the \emph{permanent} of the biadjacency matrix. This problem is well understood for general graphs, at least from a computational complexity viewpoint. Exactly counting perfect matchings has long been known to be \numP-complete~\cite{Valian79a}, and this remains true even for graphs of maximum degree~3~\cite{DagLub92}. The problem is well known to be in \FP for planar graphs~\cite{Kastel63}. For other graph classes, less is known, but \numP-completeness is known for chordal and chordal bipartite graphs~\cite{OkUeUn10}. In section~\ref{sec:exact}, we give positive results for three graph classes. Definitions and relationships between the classes we study are given in the Appendix. See also~\cite{BrLeSp99} and~\cite{Ridder}. The Appendix also gives a convenient summary of results.

Approximate counting of perfect matchings is known to be in randomized polynomial time for bipartite graphs~\cite{JeSiVi04}, but the complexity remains open for nonbipartite graphs. The algorithm of~\cite{JeSiVi04} is remarkable, but complex. It involves repeatedly running a rapidly mixing Markov chain on a sequence of edge-weighted graphs, starting from the complete bipartite graph, and gradually adapting the edge weights until they approximate the target graph. Simpler methods have been proposed, but do not lead to polynomial time approximation algorithms in general.

In~\cite{DyJeMu17}, we studied a particular Markov chain on perfect matchings in a graph, the \emph{switch chain}, on some \emph{hereditary} graph classes. That is, classes of graphs for which any vertex-induced subgraph of a graph in the class is also in the class. For reasons given in~\cite{DyJeMu17}, we believe that hereditary classes are the appropriate objects of study in this context.  For the switch chain, we asked: for which hereditary classes is the Markov chain ergodic and for which is it rapidly mixing? We provided a precise answer to the ergodicity question and close bounds on the mixing question. In particular, we showed that the mixing time of the switch chain is polynomial for the class of \emph{monotone graphs}~\cite{DiGrHo01} (also known as \emph{bipartite permutation graphs}~\cite{SpBrSt87} and \emph{proper interval bigraphs}~\cite{HelHua04}).

In this paper, we extend the analysis of~\cite{DyJeMu17} to hereditary classes of nonbipartite graphs. In section~\ref{sec:ergodicity} we consider the question of ergodicity, and in section~\ref{sec:quasimon} we consider rapid mixing of the switch chain. In both cases, we introduce corresponding new graph classes, and examine their relationship to known classes. In section~\ref{sec:slowmixing}, we show that the switch chain can have exponential mixing time in some well known hereditary graph classes.

Some reasons for restricting attention to hereditary classes are given in~\cite{DyJeMu17}. However, we might consider the class of \emph{all} graphs on which the switch chain is ergodic. In section~\ref{sec:generalswitch}, we discuss the question of deciding ergodicity of the switch chain for an arbitrary graph. We give no definitive answer, but give some evidence that polynomial time recognition is unlikely.

Finally, in section~\ref{sec:exact}, we give positive results for \emph{exactly} counting perfect matchings in some ``small'' graph classes.

\subsection{Notation and definitions}\label{sec:notation}
Let $\IN=\{1,2,\ldots\}$ denote the natural numbers, and $\IN_0=\IN\cup\{0\}$. If $n\in\IN$, let $[n]=\{1,2,\ldots,n\}$.
% and, if $n_1,n_2\in\IN_0$, let $[n_1,n_2]=\{n_1,n_1+1,\ldots,n_2\}$.
For a set $S$, $S^{(2)}$ will denote the set of subsets of $V$ of size exactly 2. For a singleton set, we will generally omit the braces, where there is no ambiguity. Thus, for example $S\cup x$ will mean $S\cup \{x\}$.

Let $G=(V,E)$ be a (simple, undirected) graph with $|V|=n$. More generally, if $H$ is any graph, we denote its vertex set by $V(H)$, and its edge set by $E(H)$. We write an $e\in E$ between $v$ and $w$ in $G$ as $e=vw$, or $e=\{v,w\}$ if the $vw$ notation is ambiguous. The degree of a vertex $v\in V$ will be denoted by \deg{v}, and its neighbourhood by \nbh{v}.

The \emph{empty graph} $G=(\emptyset,\emptyset)$ is the unique graph with $n=0$. We include the empty graph in the class of connected graphs.  Also, $G=(V,V^{(2)})$, is the complete graph on $n$ vertices. The \emph{complement} of any graph $G=(V,E)$ is $\ol{G}=(V,V^{(2)} \sm E)$.

If $U\subseteq V$, we will write $G[U]$ for the subgraph of $G$ induced by $U$.
Then a class $\cC$ of graphs is called \emph{hereditary} if $G[U]\in\cC$ for all $G\in\cC$ and $U\subseteq V$. For a cycle $C$ in $G$, we will write $G[C]$ as shorthand for $G[V(C)]$.
Definitions of the hereditary graph classes we consider, and relationships between them, are given in the Appendix.

Let $L,R$ be a \emph{bipartition} of $V$, i.e. $V=L\cup R$, $L\cap R=\emptyset$. Then we will denote the $L,R$ cut-set by $\LR=\{vw\in E: v\in L, w\in R\}$. The associated bipartite graph $(L\cup R,\LR)$ will be denoted by $G[\LR]$. If $C$ is an even cycle in $G$, then an \emph{alternating bipartition} of $C$ assigns the vertices of $C$ alternately to $L$ and $R$ as the cycle is traversed.

A \emph{matching} $M$ is an independent set of  edges in $G$. That is $M\subseteq E$, and $e\cap e'=\emptyset$ for all $\{e,e'\}\in M^{(2)}$. A \emph{perfect} matching $M$ is such that, for every $v\in V$, $v\in e$ for some $e\in M$. For a perfect matching $M$ to exist, it is clearly necessary, but not sufficient, that $n$ is even. Then $|M|=n/2$. A \emph{near-perfect} matching $M'$ is one with $|M'|=n/2-1$.  The empty graph has the unique perfect matching $\emptyset$.

A \emph{hole} in a graph $G$ will mean a chordless cycle of length greater than 4, as in e.g.~\cite[Definition~1.1.4]{BrLeSp99}. Note that the term has been also used to mean a chordless
cycle of length at least 4, as in e.g.~\cite{Vuskov10}.

\subsection{Approximate counting and the switch chain}\label{sec:switch}

Sampling a perfect matching almost uniformly at random from a graph $G=(V,E)$ is known to be computationally equivalent to approximately counting all perfect matchings~\cite{JeVaVa86}. The approximate counting problem was considered by Jerrum and Sinclair~\cite{JerSin89}, using a Markov chain similar to that suggested by Broder~\cite{Broder}. They showed that their chain has polynomial time convergence if ratio of the number of near-perfect matchings to the number of perfect matchings in the graph is polynomially bounded as a function of $n$. They called graphs with this property \emph{P-stable}, and it was investigated in~\cite{JeSiMc92}. However, many simple classes of graphs fail to have this property (e.g.~chain graphs; in the Appendix, we indicate which of the classes we consider are P-stable.) A further difficulty with this algorithm is that the chain will usually produce only a near-perfect matching, and may require many repetitions before it produces a perfect matching.

For any \emph{bipartite} graph, the Jerrum, Sinclair and Vigoda algorithm~\cite{JeSiVi04} referred to above gives polynomial time approximate counting of all perfect matchings. This is a major theoretical achievement, though the algorithm seems too complicated to be used in practice. Moreover, the approach does not appear to extend to nonbipartite graphs, since odd cycles are problematic.

For this reason, a simpler Markov chain was proposed in~\cite{DiGrHo01}, which was called the \emph{switch chain} in~\cite{DyJeMu17}. This mixes rapidly in cases that the Jerrum-Sinclair chain does not, and vice versa, so the two cannot be compared. For a graph $G$ possessing some perfect matching $M_0$, the switch chain maintains a perfect matching $M_t$ for each $t \in [t_{\max}]$, whereas the Jerrum-Sinclair chain does not. It may be described as follows.
% For each $t \in \IN_0$ let $M_t$ be perfect matching at time $t$.

\begin{enumerate}[label={(\arabic*)}]
  \item[ ] \textbf{Switch chain}
% \item[] Let the perfect matching at time $t$ be $M_t$.
  \setcounter{enumi}{0}
  \item Set $t\gets0$, and find any perfect matching $M_0$ in $G$.
  \item \label{chain:step2}
  Choose $v,v'\in V$, uniformly at random. Let $u,u'\in V$ be such that $uv,\,u'v'\in M_t$.
  \item \label{chain:step3}
  If $u'v,\,uv'\in E$, set $M_{t+1} \gets \{u'v,uv'\} \cup M_t \sm \{uv,u'v'\}$.
  \item Otherwise, set $M_{t+1}\gets M_t$.
  \item Set $t\gets t+1$. If $t<t_{\max}$, repeat from step~\ref{chain:step2}. Otherwise, stop.
\end{enumerate}

A transition of the chain is called a \emph{switch}. This chain is clearly symmetric on the set of perfect matchings, and hence will converge to the uniform distribution on perfect matchings, provided the chain is ergodic. It is clearly aperiodic, since there is delay probability of at least $1/n$ at each step, from choosing $v=v'$ in step~\ref{chain:step2}. For any $v\neq v'$ the transition probability is at most $2/n^2$, since the choice $v,v'$ can also appear as $v',v$, but the transition may fail in step~\ref{chain:step3}.

\section{Ergodicity of the switch chain}\label{sec:ergodicity}

For a graph $G$, we define the \emph{transition graph} \G{G} of the switch chain on $G$ as having a vertex for each perfect matching $M$ in $G$, and an edge between every two perfect matchings $M$, $M'$ which differ by a single switch. Then we will say $G$ is \emph{ergodic} if $\G{G}$ is connected. Since the switch chain is aperiodic, this corresponds to the usual definition of ergodicity when \G{G} is non-empty.  A class $\cC$ of graphs will be called ergodic if every $G\in\cC$ is ergodic.

As in~\cite{DyJeMu17}, we say that a graph $G=(V,E)$ is \emph{hereditarily ergodic} if, for every $U\subseteq V$, the induced subgraph $G[U]$ is ergodic.  As discussed in \cite{DyJeMu17}, this notion is closely related to that of \emph{self-reducibility} (see, for example~\cite{JerSin89}).
A class  of graphs $\cC$  will be called hereditarily ergodic if every $G\in\cC$ is hereditarily ergodic. We characterise the class of all hereditarily ergodic graphs below. This is the largest hereditary subclass of the (non-hereditary) class of ergodic~graphs.

If $\G{G}$ is the empty graph, then $G$ is ergodic. If $G$ has a unique perfect matching, $G$ is ergodic, since $\G{G}$ has a single vertex, and so is connected. Otherwise, let $X$, $Y$ be any two distinct perfect matchings in $G=(V,E)$. Then $X$ is connected to $Y$ in \G{G} if there is a sequence of switches in $G$ which \emph{transforms} $X$ to $Y$.  Since, $X\oplus Y$ is a set of vertex-disjoint alternating even cycles, it suffices to transform $X$ to $Y$ independently in each of these cycles. Thus, since we are dealing with a hereditary class, we must be able to transform $X$ to $Y$ in the graph induced in $G$ by every even cycle. Therefore, it is sufficient to decide whether or not  we can transform $X$ to $Y$ when $X\cup Y$ is an alternating Hamilton cycle in $G$.

Thus, let $H$: $v_1\to v_2\to\ldots\to v_{2r}\to v_1$ be a Hamilton cycle in the graph $G=(V,E)$, where $V=\{v_1,v_2,\ldots ,v_{2r}\}$. Let $X$, $Y$ be the two perfect matchings which form $H$ and suppose, without loss of generality, that $X=\{v_{2i-1}v_{2i}:i\in[r]\}$. Now the first switch in a sequence from $X$ to $Y$ must use a 4-cycle in $G$ with two edges  $v_{2i-1}v_{2i},\, v_{2j-1}v_{2j}\in X$, with $1\leq i< j\leq r$. The other two edges of the cycle must be either $v_{2i-1}v_{2j}$, $v_{2i}v_{2j-1}$ or $v_{2i-1}v_{2j-1}$, $v_{2i}v_{2j}$. We call the first an \emph{odd} switch, and the second an \emph{even} switch, see Fig.~\ref{fig060}.

The only switch that can change an edge in $X$ to an edge in $Y$ must have $v_{2i}v_{2j-1}\in Y$, and hence $j=i+1$. We will call this a \emph{boundary} switch. Clearly a boundary switch is an odd switch, see Fig.~\ref{fig060}.
\tikzset{vertex/.style={circle,draw,inner sep=1pt}}
\begin{figure}[ht]
\centering{%
\scalebox{0.8}{\begin{tikzpicture}[yscale=1.5,xscale=3]
\draw (0,0) node[vertex,label=left:$v_{2i-1}$] (2i-1){} (0,0.5) node[vertex,label=left:$v_{2i}$] (2i){}
(1,0.5) node[vertex,label=right:$v_{2j-1}$] (2j-1){} (1,0) node[vertex,label=right:$v_{2j}$] (2j){} ;
\draw (2i-1) .. controls (0,-1.5) and (1,-1.5).. (2j);
\draw (2i) .. controls (0,2) and (1,2).. (2j-1);
\draw (2i-1)--(2i) (2j-1)--(2j) ;
\draw[densely dotted,line width=1pt] (2i-1)--(2j) (2j-1)--(2i) ;
\draw (0.9,1.5) node[empty] {$P_1$}  (0.95,-0.85) node[empty] {$P_2$} ;
\end{tikzpicture}}\hspace*{1cm}
\scalebox{0.8}{\begin{tikzpicture}[yscale=1.5,xscale=3]
\draw (0,0) node[vertex,label=left:$v_{2i-1}$] (2i-1){} (0,0.5) node[vertex,label=left:$v_{2i}$] (2i){}
(1,0.5) node[vertex,label=right:$v_{2j-1}$] (2j-1){} (1,0) node[vertex,label=right:$v_{2j}$] (2j){} ;
\draw (2i-1) .. controls (0,-1.5) and (1,-1.5).. (2j);
\draw (2i) .. controls (0,2) and (1,2).. (2j-1);
\draw (2i-1)--(2i) (2j-1)--(2j) ;
\draw[densely dotted,line width=1pt]  (2i-1)--(2j-1) (2i)--(2j) ;
\draw (0.9,1.5) node[empty] {$P_1$}  (0.95,-0.85) node[empty] {$P_2$} ;
\end{tikzpicture}}\hspace*{1cm}
\raisebox{-1.5mm}{\scalebox{0.8}{\begin{tikzpicture}[yscale=2.2,xscale=3.1]
\draw (0,0) node[vertex,label=left:$v_{2i-1}$] (2i-1){} (0.2,0.5) node[vertex,label=left:$v_{2i}$] (2i){}
(0.8,0.5) node[vertex,label=right:$v_{2i+1}$] (2i+1){} (1,0) node[vertex,label=right:$v_{2i+2}$] (2i+2){} ;
\draw (2i-1) .. controls (0,-1.5) and (1,-1.5).. (2i+2);
\draw (2i-1)--(2i)--(2i+1)--(2i+2) ;
\draw[densely dotted,line width=1pt]  (2i-1)--(2i+2) ;
\draw  (0.5,0.625) node[empty] {$P_1$} (0.95,-0.85) node[empty] {$P_2$} ;
\end{tikzpicture}}}}\vspace{1ex}
\caption{An odd switch, an even switch and a boundary switch}\label{fig060}
\end{figure}
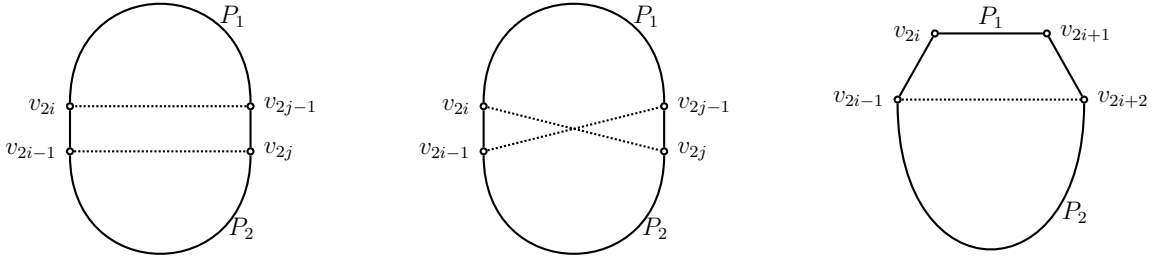

The edges $v_{2i-1}v_{2i},\,v_{2j-1}v_{2j}$ divide $H$ into two vertex-disjoint paths $P_1:v_{2i}\to v_{2j-1}$ and $P_2:v_{2j}\to v_{2i-1}$, see Fig.~\ref{fig060}.  Thus performing an odd switch on a 4-cycle $(v_{2i-1},v_{2i},v_{2j-1},v_{2j})$ results in two smaller alternating cycles $P_1\cup v_{2i-1}v_{2j}$ and $P_2\cup v_{2i}v_{2j-1}$, on which we can use induction, since we are in a hereditary class. However, performing an even switch on a 4-cycle $(v_{2i-1},v_{2j-1},v_{2i},v_{2j})$ simply produces a new Hamilton cycle
$H'=X'\cup Y$, where $X'= X\sm\{v_{2i-1}v_{2i},\, v_{2j-1}v_{2j}\}\cup\{v_{2i-1}v_{2j-1},\, v_{2i}v_{2j}\}$, see Fig.~\ref{fig060}.

An edge $v_iv_j\in E\,\sm\, C $ is a \emph{chord} of a cycle $C$. If $C$ is an even cycle, it is an \emph{odd} chord if $j-i=1\pmod 2$ and \emph{even} if $j-i=0\pmod 2$. Note that $j-i=i-j\pmod 2$, so the definition is independent of the order of $i$ and $j$ on $C$.   Note that even and odd chords are not defined for odd cycles.

An odd chord divides an even cycle $C$ into two even cycles, sharing an edge. Thus an odd switch involves two odd chords, and an even switch involves two even chords. However, a 4-cycle with two odd chords may not be an odd switch and a 4-cycle with two even chords may not be an even switch, if the cycle edges involved are not both in $X$ or $Y$. We call these \emph{illegal} switches, see Fig.\ref{fig070}.
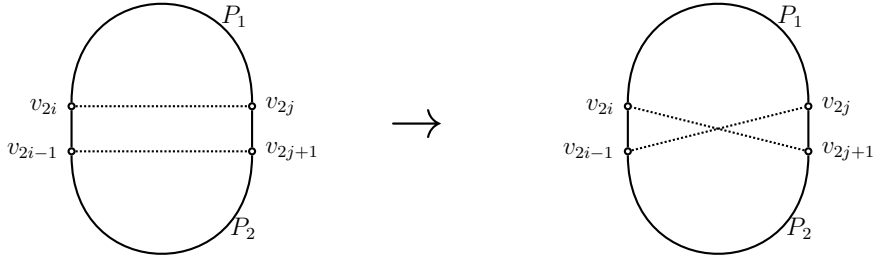
\begin{figure}[H]
\centering{%
\scalebox{0.8}{\begin{tikzpicture}[yscale=1.5,xscale=3]
\draw (0,0) node[vertex,label=left:$v_{2i-1}$] (2i-1){} (0,0.5) node[vertex,label=left:$v_{2i}$] (2i){}
(1,0.5) node[vertex,label=right:$v_{2j}$] (2j){} (1,0) node[vertex,label=right:$v_{2j+1}$] (2j+1){} ;
\draw (2i-1) .. controls (0,-1.5) and (1,-1.5).. (2j+1);
\draw (2i) .. controls (0,2) and (1,2).. (2j);
\draw (2i-1)--(2i) (2j)--(2j+1) ;
\draw[densely dotted,line width=1pt] (2i)--(2j) (2i-1)--(2j+1) ;
\draw (0.9,1.5) node[empty] {$P_1$}  (0.95,-0.85) node[empty] {$P_2$} ;
\end{tikzpicture}}\raisebox{2cm}{\qquad\LARGE$\rightarrow$\qquad}
\scalebox{0.8}{\begin{tikzpicture}[yscale=1.5,xscale=3]
\draw (0,0) node[vertex,label=left:$v_{2i-1}$] (2i-1){} (0,0.5) node[vertex,label=left:$v_{2i}$] (2i){}
(1,0.5) node[vertex,label=right:$v_{2j}$] (2j){} (1,0) node[vertex,label=right:$v_{2j+1}$] (2j+1){} ;
\draw (2i-1) .. controls (0,-1.5) and (1,-1.5).. (2j+1);
\draw (2i) .. controls (0,2) and (1,2).. (2j);
\draw (2i-1)--(2i) (2j+1)--(2j) ;
\draw[densely dotted,line width=1pt]  (2i-1)--(2j) (2j+1)--(2i) ;
\draw (0.9,1.5) node[empty] {$P_1$}  (0.95,-0.85) node[empty] {$P_2$} ;
\end{tikzpicture}}}\vspace{1ex}
\caption{An illegal switch}\label{fig070}
\end{figure}

We  define the graph class \och\ as follows. A  graph $G=(V,E)$ is odd chordal if and only if every even cycle $C$ in $G$ of length six or more has an odd chord. Note that this is  a hereditary graph property.  The switch chain is hereditarily ergodic on the class \och, but it is not the largest class with this property.

Let  $(v_{2i-1},v_{2i}, v_{2j},v_{2j-1})$ be an even switch for the even cycle $C$, with cycle segments $P_1$, $P_2$ as above. Then a \emph{crossing chord} is an edge $(v_k,v_l)$ such that $v_k\in P_1$, $v_l\in P_2$, see Fig.~\ref{fig040}.

We  can now define our target graph class \switch.  A graph $G=(V,E)$ is switchable if and only if every even cycle $C$ in $G$ of length 6 or more has an odd chord, or has an even switch with a crossing chord. Clearly we may assume that the crossing chord is an even chord. This class is also hereditary, and the definition implies $\och\subseteq\switch$.

Our choice of names for the classes \och\ and \switch\ is obvious from the above, and Theorem~\ref{thm:ergodic} below.
\begin{figure}[ht]
\begin{center}
\scalebox{0.8}{\begin{tikzpicture}[yscale=1.5,xscale=3]
\draw (0,0) node[vertex,label=left:$v_{2i-1}$] (2i-1){} (0,0.5) node[vertex,label=left:$v_{2i}$] (2i){}
(1,0.5) node[vertex,label=right:$v_{2j-1}$] (2j-1){} (1,0) node[vertex,label=right:$v_{2j}$] (2j){} ;
\draw (2i-1) .. controls (0,-1.5) and (1,-1.5).. (2j);
\draw (2i) .. controls (0,2) and (1,2).. (2j-1);
\draw (2i-1)--(2i) (2j-1)--(2j) ;
\draw[densely dotted,line width=1pt] (2i)--(2j) (2i-1)--(2j-1) ;
\draw (0.25,1.5) node[vertex,fill=white,label=above:$v_k$] (k) {} (0.25,-1) node[vertex,fill=white,label=below:$v_{l}$] (l) {};
\draw[densely dotted,line width=1pt]  (k)--(l) ;
\draw (0.9,1.5) node[empty] {$C_1$}  (0.95,-0.9) node[empty] {$C_2$} ;
\end{tikzpicture}}\raisebox{2cm}{\qquad\LARGE$\rightarrow$\qquad}
\scalebox{0.8}{\begin{tikzpicture}[yscale=1.5,xscale=3]
\draw (0,0) node[vertex,label=left:$v_{2i-1}$] (2i-1){} (0,0.5) node[vertex,label=left:$v_{2j-1}$] (2j-1){}
(1,0.5) node[vertex,label=right:$v_{2i}$] (2i){} (1,0) node[vertex,label=right:$v_{2j}$] (2j){} ;
\draw (2i-1) .. controls (0,-1.5) and (1,-1.5).. (2j);
\draw (2j-1) .. controls (0,2) and (1,2).. (2i);
\draw (2j)--(2i) (2i-1)--(2j-1) ;
\draw[densely dotted,line width=1pt] (2i)--(2i-1) (2j-1)--(2j) ;
\draw (0.85,1.35) node[vertex,fill=white,label=above:$v_k$] (k) {} (0.3,-1.05) node[vertex,fill=white,label=below:$v_{l}$] (l) {};
\draw[densely dotted,line width=1pt]  (k)--(l) ;
\draw (0.1,1.55) node[empty] {$C_1$}  (0.95,-0.9) node[empty] {$C_2$} ;
\end{tikzpicture}}
\end{center}
\caption{An even switch with a crossing chord}\label{fig040}
\end{figure}
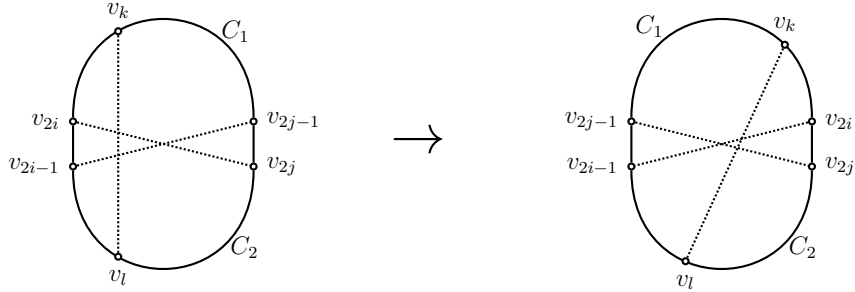

\begin{theorem}\label{thm:ergodic}
 A graph  $G=(V,E)$ is hereditarily ergodic if and only if $G\in\switch$.
\end{theorem}
\begin{proof}

Suppose $G\in\switch$\ and $G$ has a Hamilton cycle $H$ which is the union of two perfect matchings $X$ and $Y$. We wish to show that $X$ can be transformed to $Y$ using switches in $G$.
We will argue inductively on the size of $G$. If $H$ is a 4-cycle, we can transform $X$ to $Y$ with a single switch. Suppose then that we can transform $X'$ to $Y'$ for every two perfect matchings $X', Y'$ in any graph $G'\in\switch$\ that has fewer vertices than $G$.
% Now $G$ is near~odd~chordal, so
First suppose $C$ has an odd chord $(v_i,v_j)$. Then $H\cup\thsp v_iv_j$ forms two even cycles $C_1$, $C_2$, with $v_iv_j$ as a common edge, so that $v_{i+1}\in C_1$ and $v_{i-1}\in C_2$, see Fig.~\ref{fig050}. If $i$ is odd and $j$ is even, then $v_{i}v_{i+1},\,v_{j-1}v_j\in C_1\cap X$, and if $i$ is even and $j$ is odd, then $v_{i-1}v_i,\,v_jv_{j+1}\in C_2\cap X$. In the first case $C_1$ is an alternating cycle for $X'=X\cap C_1$, $Y'=C_1\sm X'$, and in the second $C_2$ is an alternating cycle for $X'=X\cap C_2$, $Y'=C_2\sm X'$. Consider the first case, the second being symmetrical. Then, since $C_1$ is shorter than $H$, we can transform $X'$ to $Y'$ by induction. After this, $C_2$ is an alternating cycle shorter than $C$, with $X''=(X\cap C_2)\cup\thsp v_iv_j$, $Y''=Y\cap C_2$, so we can transform $X''$ to $Y''$ by induction. This transforms $X$ to $Y$ for the whole cycle $H$, and we are done.

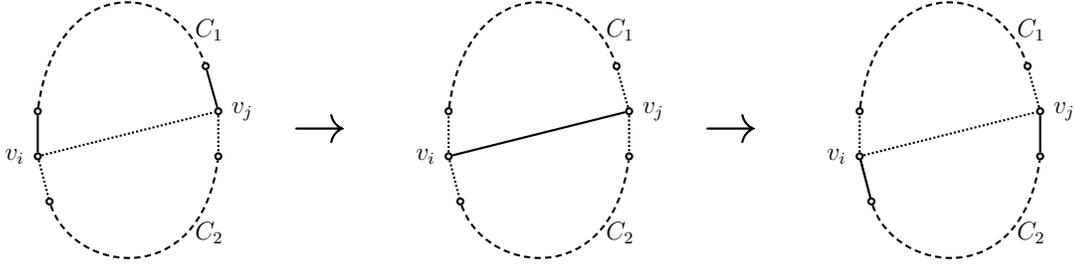
\begin{figure}[H]
\begin{center}
\scalebox{0.8}{\begin{tikzpicture}[yscale=1.5,xscale=3]
\draw (0,0) node[vertex,label=left:$v_i$] (vi){} (0,0.5) node[vertex] (vi'){}
(1,0.5) node[vertex,label=right:$v_j$] (vj){} (1,0) node[vertex] (vj'){}
(0.93,1) node[vertex] (vj''){} (0.065,-0.5) node[vertex] (vi''){};
\draw[densely dashed] (vi') .. controls (0.1,2) and (0.75,2) .. (vj'');
\draw[densely dashed] (vi'') .. controls (0.2,-1.3) and (0.9,-1.5) .. (vj');
\draw[densely dotted] (vi)--(vi'') (vj)--(vj') (vi)--(vj) ;
\draw (vi)--(vi') (vj)--(vj'') ;
\draw (0.95,1.4) node[empty] {$C_1$}  (0.95,-0.85) node[empty] {$C_2$} ;
\end{tikzpicture}}\raisebox{2cm}{\quad\LARGE$\rightarrow$\quad}
\scalebox{0.8}{\begin{tikzpicture}[yscale=1.5,xscale=3]
\draw (0,0) node[vertex,label=left:$v_i$] (vi){} (0,0.5) node[vertex] (vi'){}
(1,0.5) node[vertex,label=right:$v_j$] (vj){} (1,0) node[vertex] (vj'){}
(0.93,1) node[vertex] (vj''){} (0.065,-0.5) node[vertex] (vi''){};
\draw[densely dashed] (vi') .. controls (0.1,2) and (0.75,2) .. (vj'');
\draw[densely dashed] (vi'') .. controls (0.2,-1.3) and (0.9,-1.5) .. (vj');
\draw[densely dotted] (vi)--(vi'') (vj)--(vj') (vi)--(vi') (vj)--(vj'') ;
\draw (vi)--(vj) ;
\draw (0.95,1.4) node[empty] {$C_1$}  (0.95,-0.85) node[empty] {$C_2$} ;
\end{tikzpicture}}\raisebox{2cm}{\quad\LARGE$\rightarrow$\quad}
\scalebox{0.8}{\begin{tikzpicture}[yscale=1.5,xscale=3]
\draw (0,0) node[vertex,label=left:$v_i$] (vi){} (0,0.5) node[vertex] (vi'){}
(1,0.5) node[vertex,label=right:$v_j$] (vj){} (1,0) node[vertex] (vj'){}
(0.93,1) node[vertex] (vj''){} (0.065,-0.5) node[vertex] (vi''){};
\draw[densely dashed] (vi') .. controls (0.1,2) and (0.75,2) .. (vj'');
\draw[densely dashed] (vi'') .. controls (0.2,-1.3) and (0.9,-1.5) .. (vj');
\draw[densely dotted] (vi)--(vi'') (vi)--(vj) (vi)--(vi') (vj)--(vj'') ;
\draw  (vi)--(vi'') (vj)--(vj');
\draw (0.95,1.4) node[empty] {$C_1$}  (0.95,-0.85) node[empty] {$C_2$} ;
\end{tikzpicture}}
\end{center}
\caption{Switching a cycle using an odd chord}\label{fig050}
\end{figure}

Now suppose $H$ has no odd chord, so it has an even switch with an even crossing chord. The even switch
gives another Hamilton cycle $H'=P_1\cup v_{2j-1}v_{2i-1}\cup P_2\cup v_{2j}v_{2i}$, and suppose its vertices are numbered in the implied order. Now, in this numbering, the parity of vertices in $P_1$ remains as in $C$, but the edge $v_{2j-1}v_{2i-1}$ changes the parity of all vertices in $P_2$. Finally, the edge $v_{2j}v_{2i}$ restores the parity in $P_1$. Thus, in particular the crossing chord $v_kv_l$ changes from being an even chord in $H$ to an odd chord in $H'$. Now, since  $H'$ has an odd chord, we can use the argument above to show that its matching $X'=X\sm\{v_{2i-1}v_{2i},\,v_{2i-1}v_{2j}\}\cup\{v_{2i-1}v_{2j-1},\,v_{2i}v_{2j}\}$ can be transformed to $Y$.

Suppose $G\notin\switch$. Then we may assume that there is a Hamilton cycle $H_0$ in $G$, of length $2r\geq 6$, which has only even chords. Let $X_0$, $Y$ be the two perfect matchings such that $H_0=X_0\cup Y$. Then $H_0$ has only even switches, and no even switch can have a crossing chord.  Therefore, any switch from $X_0$ to $X_1$ in $G$ produces a new Hamilton cycle $H_1=X_1\cup Y$ in $G$. Since there are no crossing chords, the switch does not change the parity of any chord from $H_0$ to $H_1$, so $H_1$ also has only even chords, and hence only even switches. The switch cannot produce a crossing chord for any even switch in $H_1$, since this chord would also have been crossing in $H_0$.
Thus $H_1$ is a Hamilton cycle with no odd chord and no even switch with a crossing chord.

Therefore, suppose there is a sequence of switches, $X_0,X_1,\ldots,X_i,\ldots,X_l=Y$. Let $r$ be the smallest $i$ such that $X_i\cap Y\neq\emptyset$. The switch from $X_{r-1}$ to $X_r$ introduces an edge of $Y$, and so requires a boundary switch in $H_{r-1}$, which is odd switch. However, by induction, no Hamilton cycle in the sequence $H_0,H_1,\ldots,H_{r-1}$ has an odd chord, and so there can be no odd switch in $H_{r-1}$.
Hence $X_r\cap Y\neq \emptyset$, a contradiction. So the switch chain is not ergodic on $G$, as required.
\end{proof}
\begin{lemma}\label{lem:diamswitch}
If $G\in\switch$, the diameter of \G{G} is at most $(n-3)$.
\end{lemma}
\begin{proof}
Remember $n=|V(G)|$. Let $D_n$ be the diameter of \G{G}. Clearly $D_4=1=4-3$, and this will be the basis for an induction. Using the construction in the proof of Theorem~\ref{thm:ergodic},  the graph $G$ is decomposable into
  two smaller graphs $G_1$ and $G_2$, which have a common edge, after one switch. Let $G_1$ and $G_2$  have $m+1$ and $n-m+1$ vertices, for some $m$. Thus by induction, $D_n\leq D_{m+1}+D_{n-m+1}+1\leq (m-2)+(n-m-2)+1=n-3$.
\end{proof}
For the class \och, we can prove a stronger bound, which also gives a characterisation of the class in terms of the switch chain,
\begin{lemma}\label{lem:oddchordwitch}
$G\in\och$ if and only if $\diam{\G{C}}=|C|/2-1$ for every even cycle $C$ in $G$.
\end{lemma}
\begin{proof}
  Let $C$ be any even cycle in $G\in\och$. Then $C$ has a boundary switch. First consider the matchings $X,Y$ for which $C$ is an alternating cycle. To switch $X$ to $Y$, we first perform the boundary switch, leaving an alternating cycle $C'$ with $|C'|=|C|-2$, Assume by induction that $\dist(X,Y)=|C|/2-1$, the base case being for a quadrangle, here $\dist(X,Y)=1=4/2-1$.   Then $\dist(X,Y)=1+(|C'|/2-1)=(|C|-2)/2=|C|/2-1$, continuing the induction. Thus $\diam{\G{C}}\geq|C|/2-1$. Now, if $X,Y$ are any two matching in $\G{C}$, $X\oplus Y$ can be divided into alternating cycles $C_1,C_2,\ldots,C_k$, say.   Then $\dist(X,Y)\leq\sum_{i=1}^{k}(|C_i|/2-1)=|C|/2-k\leq |C|/2-1$, with equality if and only if $X\cup Y=C$. Thus $\diam{\G{C}}=|C|/2-1$.

  Conversely, suppose $C$ is an even cycle with no odd chord, but $\G{C}\in\switch$, so $\diam{\G{C}}$ is well defined.  First consider matchings $X,Y$ such that $C$ is an alternating cycle. Then the first switch on the path from $X$ to $Y$ must be an even switch, giving matchings $X'$, $Y'$ which form an alternating cycle $C'$ with $|C'|=|C|$. Now, from above, $\dist(X',Y')\geq |C'|/2-1=|C|/2-1$. and thus $\dist(X,Y)\geq |C|/2$. Hence $\diam{\G{C}}\neq|C|/2-1$.
\end{proof}
\begin{corollary}\label{lem:diamoddchord}
If $G=(V,E)\in\och$, with $n=|V|$, then $\diam{\G{G}}\leq n/2-1$.
\end{corollary}
\begin{proof}
 If $X,Y$ are any two matching in $\G{G}$, $X\oplus Y$ can be divided into alternating cycles $C_1,C_2,\ldots,C_k$, say.   Then $\dist(X,Y)\leq\sum_{i=1}^{k}(|C_i|/2-1)\leq n/2-k\leq n/2-1$.
\end{proof}
Finally, we note that, even if the switch chain is not ergodic on a graph $G$, it may still be able to access an exponential number of perfect matchings from any given perfect matching. Thus the graphs which are not ergodic for the switch chain do not necessarily have ``frozen'' perfect matchings. Since the existence of frozen states is the most usual criterion for non-ergodicity of large Markov chains, deciding non-hereditary ergodicity seems problematic.
\begin{example}\label{example20}
The graph $G$ in Fig.~\ref{fig080} has $n=4k$ vertices, and a Hamilton cycle $H=X\cup Y$, where $X$, $Y$ are the following two perfect matchings:
\[X=\big\{\{1,2\},\{3,4\},\ldots,\{4k-1,4k\}\big\},\quad Y=\big\{\{2,3\},\{4,5\},\ldots,\{4k-2,4k-1\},\{4k,1\}\big\}.\]
From either $X$ or $Y$, there are $k$ even switches, each without a crossing chord.
Each of these switches can be made independently, leading to $2^k=2^{n/4}$ different matchings.
However $Y$ cannot be reached from $X$, or vice versa. Note that there are $4(k-1)$ illegal switches for $H$,
for example $(1,2,4k-2,4k-1)$ and $(2,3,4k-1,4k)$.
\tikzset{vertex/.style={circle,draw,fill=none,inner sep=1pt}}
\begin{figure}[H]
\begin{center}
\scalebox{0.9}{\begin{tikzpicture}[xscale=2,yscale=1.5,line width=0.8pt,font=\small]
\path (0,0) node[vertex] (1') {}  +(0,-0.2) node[empty] {$4k$}
++(1,0) node[vertex] (2) {}  +(0,-0.2) node[empty] {$4k-1$}
++(1,0) node[vertex] (3') {}  +(0,-0.2) node[empty] {$4k-2$}
++(0.5,0) coordinate (u')
++(1,0) node[vertex] (4) {}  +(0,-0.2) node[empty] {$2k+2$}
++(1,0) node[vertex] (n') {}  +(0,-0.16) node[empty] {$2k+1$}
++(0,1) node[vertex] (n) {}  +(0,0.2) node[empty] {$2k$}
++(-1,0) node[vertex] (4') {}  +(0,0.2) node[empty] {$2k-1$}
++(-0.5,0) coordinate (u) {}
++(-1,0) node[vertex] (3) {}  +(0,0.2) node[empty] {$3$}
++(-1,0) node[vertex] (2') {}  +(0,0.2) node[empty] {$2$}
++(-1,0) node[vertex] (1)  {}  +(0,0.2) node[empty] {$1$} ;
\draw (1)--(2)--(3) (1')--(2')--(3') (1')--(2) (2')--(3) (n)--(n')
(1)--(1') (1)--(2') (2)--(3') (4')--(n) (4)--(n') (4')--(n') (4)--(n);
\draw[dashed] (3')--(4) (3)--(4') ;
\end{tikzpicture}}
\end{center}
\caption{\emph{Example}~\ref{example20}}\label{fig080}
\end{figure}
\end{example}

\subsection{Relationship to known graph classes}
We will now consider the relationship between the classes defined above and known hereditary graph classes, which are defined in the Appendix.
% for definitions of these classes.
First we will show that \och$\ \subset\ $\switch, by means of the following example.
\begin{example}\label{example10}
The graph $G$ in Fig.~\ref{fig010} has an (even) Hamilton cycle $H$ is $1\to2\to3\to4\to5\to6\to7\to8\to1$, and no odd chords, but there is a sequence of switches which transforms the perfect matching $X=\{(1,2),(3,4),(5,6),(7,8)\}$, shown in solid line, to the perfect matching $Y=\{(2,3),(4,5),(6.7),(8,1)\}$, shown dashed. Other edges of $G$ are shown dotted. The switch used to obtain the (solid) perfect matching from its predecessor is shown below each graph. The first switch is an even switch $(3,7,8,4)$ with two crossing even chords $(1,5)$, $(2,6)$. In the last step two disjoint odd switches have been made simultaneously.
\tikzset{vertex/.style={circle,draw,fill=none,inner sep=0.8pt}}
\begin{figure}[htb]
\begin{center}
\begin{tikzpicture}
\begin{scope}[scale=0.8]
\draw (0,0) node[vertex,label=left:1] (1) {} ++(90:1)  node[vertex,label=left:2] (2) {}
++(45:1)  node[vertex,label=above:3] (3) {} ++(0:1)  node[vertex,label=above:4] (4) {}
++(-45:1)  node[vertex,label=right:5] (5) {} ++(-90:1)  node[vertex,label=right:6] (6) {}
++(-135:1)  node[vertex,label=below:7] (7) {} ++(-180:1)  node[vertex,label=below:8] (8) {};
\draw (1)--(2) (3)--(4) (5)--(6) (7)--(8)  ;
\draw[densely dashed] (2)--(3) (4)--(5) (6)--(7) (8)--(1)  ;
\draw[densely dotted] (1)--(5) (2)--(6) (3)--(7) (4)--(8) ;
\draw (1.25,-2.5) node[empty] {$G$} ;
\end{scope}
\begin{scope}[scale=0.8,xshift=5cm]
\draw (0,0) node[vertex,label=left:1] (1) {} ++(90:1)  node[vertex,label=left:2] (2) {}
++(45:1)  node[vertex,label=above:3] (3) {} ++(0:1)  node[vertex,label=above:7] (4) {}
++(-45:1)  node[vertex,label=right:6] (5) {} ++(-90:1)  node[vertex,label=right:5] (6) {}
++(-135:1)  node[vertex,label=below:4] (7) {} ++(-180:1)  node[vertex,label=below:8] (8) {};
\draw (1)--(2) (3)--(4) (5)--(6) (7)--(8)  ;
\draw[densely dashed] (2)--(3) (4)--(5) (6)--(7) (8)--(1)  ;
\draw[densely dotted] (1)--(6) (2)--(5) (3)--(7) (4)--(8) ;
\draw (1.25,-2.5) node[empty] {(3,7,8,4)} ;
\end{scope}
\begin{scope}[scale=0.8,xshift=10cm]
\draw (0,0) node[vertex,label=left:1] (1) {} ++(90:1)  node[vertex,label=left:2] (2) {}
++(45:1)  node[vertex,label=above:3] (3) {} ++(0:1)  node[vertex,label=above:7] (4) {}
++(-45:1)  node[vertex,label=right:6] (5) {} ++(-90:1)  node[vertex,label=right:5] (6) {}
++(-135:1)  node[vertex,label=below:4] (7) {} ++(-180:1)  node[vertex,label=below:8] (8) {};
\draw (1)--(6)  (2)--(5) (3)--(4)  (7)--(8)  ;
\draw[densely dashed]  (2)--(3) (4)--(5) (6)--(7) (8)--(1)  ;
\draw[densely dotted] (1)--(2)  (5)--(6) (3)--(7) (4)--(8) ;
\draw (1.25,-2.5) node[empty] {(1,2,6,5)} ;
\end{scope}
\begin{scope}[scale=0.8,xshift=15cm]
\draw (0,0) node[vertex,label=left:1] (1) {} ++(90:1)  node[vertex,label=left:2] (2) {}
++(45:1)  node[vertex,label=above:3] (3) {} ++(0:1)  node[vertex,label=above:7] (4) {}
++(-45:1)  node[vertex,label=right:6] (5) {} ++(-90:1)  node[vertex,label=right:5] (6) {}
++(-135:1)  node[vertex,label=below:4] (7) {} ++(-180:1)  node[vertex,label=below:8] (8) {};
\draw  (2)--(3) (4)--(5) (6)--(7)  (8)--(1) ;
\draw[densely dotted] (1)--(2) (1)--(6) (2)--(5) (3)--(4) (5)--(6) (1)--(6) (2)--(5) (3)--(7) (7)--(8) (4)--(8) ;
\draw (1.25,-2.5) node[empty] {(2,6,7,3),\,(1,5,4,8)} ;
\end{scope}
\end{tikzpicture}
\end{center}
\caption{\emph{Example} 1.}\label{fig010}
\end{figure}
\end{example}
This example may be extended so  that the outer cycle has any even number of vertices. These graphs are the \emph{M\"obius ladders}, which appear in~\cite{DyeMue18}, in a related context.

Next we consider the simple class\cogr. We may show that $\cogr\nsubseteq\switch$, and hence the switch chain is not necessarily ergodic even in this class. Consider the graph $G$ shown in Fig.~\ref{fig090}:
\tikzset{vertex/.style={circle,draw,fill=none,inner sep=1pt}}
\begin{figure}[H]
\centering{%
\scalebox{0.9}{\begin{tikzpicture}[xscale=2,yscale=2.75,font=\small]
\path
(0,0) node[vertex,label=below:5] (5) {}
(-0.5,0.5) node[vertex,label=left:6] (6) {}
(0,1) node[vertex,label=above:1] (1) {}
(1,1) node[vertex,label=above:2] (2) {}
(1.5,0.5) node[vertex,label=right:3] (3) {}
(1,0) node[vertex,label=below:4] (4) {};
\draw (1)--(2)--(3)--(4)--(5)--(6)--(1) (3)--(1) (5)--(3) (6)--(2) (4)--(6)  ;
\end{tikzpicture}}}
\caption{A non-ergodic cograph}\label{fig090}
\end{figure}
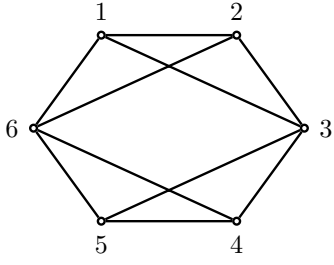

This is a cograph, since $G=G_1\Join G_2$, where $G_1=G[\{1,2,4,5\}]$, $G_2=G[\{3,6\}]$ are cographs, since $G_1\simeq(K_1\Join K_1)+(K_1\Join K_1)$ and $G_2\simeq K_1 \uplus K_1$. However $G$ is not odd chordal, since the 6-cycle $(1,2,3,4,5,6)$ spans $G$ and has only even chords: $\{1,3\}$, $\{2,6\}$, $\{3,5\}$ and $\{4,6\}$. Moreover, $G$ is not switchable, since the two even switches $\{1,3\}$, \{$2,6\}$ and  $\{3,5\}$, $\{4,6\}$ have no crossing chord.
However, we will show in section~\ref{sec:exact}, that perfect matchings in a cograph can be counted exactly, and hence a random matching can be generated, in polynomial time.

It is known that \cogr\ $\subseteq$\, \perm, the class of permutation graphs. Thus we know that \perm\ $\nsubseteq$\, \switch. However, there are permutation graphs which are not cographs and are not switchable.
Consider the graph $G$ shown in Fig.~\ref{fig090}.  It has the intersection model shown, so $G\in\perm$. However, $G$ is not a cograph, since $G[\{2,3,4,5\}]\simeq P_4$. The 6-cycle $(1,2,4,3,5,6)$ spans $G$, and has no odd chord or even switch. So, by Theorem~\ref{thm:ergodic}, $G\notin\switch$. This example can be extended to an infinite sequence of connected non-ergodic permutation graphs on $2(k+1)$ vertices $(k>1)$ which are a $2k$-ladder with a triangle at each end. These graphs also appear, in a related context in~\cite{DyeMue18}.
\begin{figure}[ht]
\centering{%
\scalebox{0.9}{\begin{tikzpicture}[xscale=2,yscale=1.5,font=\small]
\path
(0,0) node[vertex,label=below:1] (1) {}
(-0.5,0.5) node[vertex,label=left:2] (2) {}
(0,1) node[vertex,label=above:4] (3) {}
(1,1) node[vertex,label=above:3] (4) {}
(1.5,0.5) node[vertex,label=right:5] (5) {}
(1,0) node[vertex,label=below:6] (6) {};
\draw (1)--(2)--(3)--(4)--(5)--(6)--(1)--(3) (4)--(6)  ;
\end{tikzpicture}}\hspace*{2cm}
\scalebox{0.9}{\begin{tikzpicture}[xscale=1.25,yscale=1.5,font=\small]
\path
(1,1) node[vertex,label=above:1] (1) {}
(1.8,1) node[vertex,label=above:2] (2) {}
(4,1) node[vertex,label=above:4] (3) {}
(3,1) node[vertex,label=above:3] (4) {}
(5.2,1) node[vertex,label=above:5] (5) {}
(6,1) node[vertex,label=above:6] (6) {}
(4,0) node[vertex,label=below:1] (1') {}
(1.8,0) node[vertex,label=below:2] (2') {}
(1,0) node[vertex,label=below:4] (3') {}
(6,0) node[vertex,label=below:3] (4') {}
(5.2,0) node[vertex,label=below:5] (5') {}
(3,0) node[vertex,label=below:6] (6') {};
\draw (1)--(1') (2)--(2') (3)--(3') (4)--(4') (5)--(5') (6)--(6')  ;
\end{tikzpicture}}}
\caption{A non-ergodic permutation graph}\label{fig100}
\end{figure}
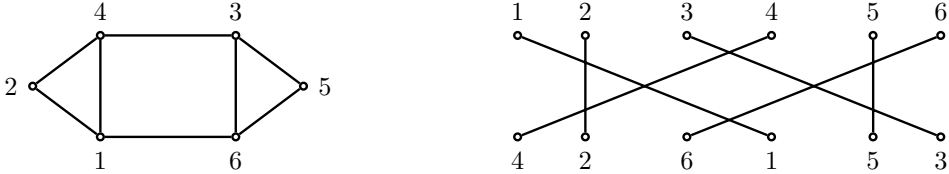

The class \och\ does not seem to have been studied previously in the graph theory literature.  We have the following relationships to the known classes \cbg, \intl, \textsc{Strongly~chordal}, \chord, \ehf\ and \och.\\[0.5ex]
\centerline{\cbg\ $\subset$\ \och\  $\subset$\ \switch\ $\subset$\ \ehf}\\[0.5ex]
\centerline{\intl\ $\subset$\ \strch\ =\ \chord$\, \cap\, $\och}\\[0.5ex]
\centerline{\cbg\ =\ \bip$\, \cap\, $\och.}\\[1ex]
The inclusions are strict, as illustrated in Example~\ref{example10} above and Fig.~\ref{fig020} below. Fig.~\ref{fig020}(a)  contains a triangle, so cannot be chordal bipartite, but has no odd hole. The only 6-cycle has an odd chord $\{1,4\}$, so the graph is odd chordal. In Fig.~\ref{fig020}(b), the 6-cycle has an even chord $\{1,5\}$, but no odd chord, so the graph has no even hole, but is not odd chordal.
In Fig.~\ref{fig020}(c), the graph is odd chordal, since the only even cycle is the 4-cycle (3,6,7,8), but it has an odd hole (1,2,3,4,5). In Fig.~\ref{fig020}(d), the outer 6-cycle has no odd chord, but is chordal, since all other cycles are triangles. So, in addition to the inclusions, we see that \chord\ and \ohf\ are  incomparable with \och. Thus, from~\cite{ChRoST06}, odd chordal graphs are not perfect in general.
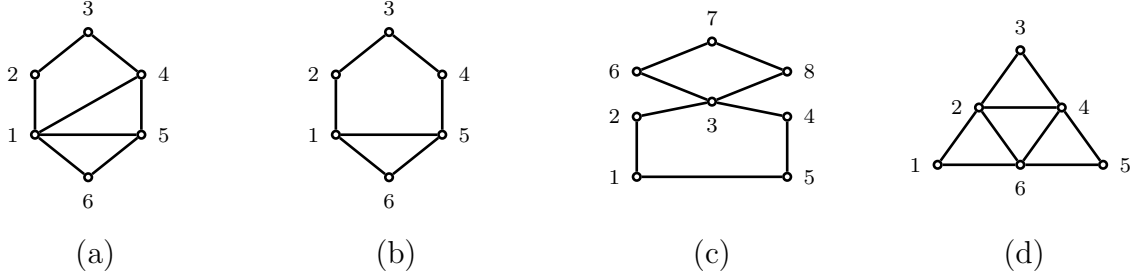
\begin{figure}[htb]
\centerline{
\begin{tikzpicture}[yscale=0.8,font=\scriptsize]
\begin{scope}
\draw (0,0) node[vertex,label=left:1] (1) {} ++(90:1)  node[vertex,label=left:2] (2) {}
++(45:1)  node[vertex,label=above:3] (3) {} ++(-45:1)  node[vertex,label=right:4] (4) {}
++(-90:1)  node[vertex,label=right:5] (5) {} ++(225:1)  node[vertex,label=below:6] (6) {};
\draw (1)--(2)--(3)--(4)--(5)--(6)--(1)--(5) (1)--(4) ;
\end{scope}
\draw (0.8,-2) node[empty] {\normalsize(a)} ;
\begin{scope}[xshift=4cm]
\draw (0,0) node[vertex,label=left:1] (1) {} ++(90:1)  node[vertex,label=left:2] (2) {}
++(45:1)  node[vertex,label=above:3] (3) {} ++(-45:1)  node[vertex,label=right:4] (4) {}
++(-90:1)  node[vertex,label=right:5] (5) {} ++(225:1)  node[vertex,label=below:6] (6) {};
\draw (1)--(2)--(3)--(4)--(5)--(6)--(1)--(5) ;
\end{scope}
\draw (4.8,-2) node[empty] {\normalsize(b)} ;
\begin{scope}[xshift=8cm,yshift=-7mm]
\draw (0,0) node[vertex,label=left:1] (1) {} ++(90:1)  node[vertex,label=left:2] (2) {}
++(1,0.25)  node[vertex,label=below:3] (3) {} ++(1,-.25)  node[vertex,label=right:4] (4) {}
++(-90:1)  node[vertex,label=right:5] (5) {}
(0,1.75)  node[vertex,label=left:6] (6) {} (1,2.25)  node[vertex,label=above:7] (7) {}(2,1.75)  node[vertex,label=right:8] (8) {};
\draw (1)--(2)--(3)--(4)--(5)--(1) (3)--(6)--(7)--(8)--(3) ;
\end{scope}
\draw (9,-2) node[empty] {\normalsize(c)} ;
\begin{scope}[xshift=12cm,yshift=-5mm]
\draw (0,0) node[vertex,label=left:1] (0) {} +(60:1.1)  node[vertex,label=left:2] (1) {} +(60:2.2)  node[vertex,label=above:3] (2) {} (1.1,0) node[vertex,label=below:6] (3) {} (2.2,0) node[vertex,label=right:5] (4) {} +(120:1.1)  node[vertex,label=right:4] (5) {};
\draw[line width=1pt] (0)--(3)--(4)--(5)--(2)--(1)--(0) (1)--(3)--(5)--(1) ;
\end{scope}
\draw (13.15,-2) node[empty] {\normalsize(d)} ;
\end{tikzpicture}
}
\caption{Distinction between graph classes}\label{fig020}
\end{figure}
Note that, except for \och, \switch\ and \ehf, all these classes are known to be recognisable in polynomial time. Thus, an obvious question is: does \och\ have a polynomial time recognition algorithm? We conjecture that the answer is ``yes'', but currently we cannot prove this.

\section{Rapid mixing and quasimonotone graphs}\label{sec:quasimon}

The switch chain for \emph{monotone} graphs (also known as \emph{bipartite permutation graphs}, or \emph{proper interval bigraphs}), was studied in~\cite{DyJeMu17}. These are graphs for which the biadjacency matrix has a ``staircase'' structure. See~\cite{DyJeMu17} for precise definitions. The chain was shown to have polynomial mixing time. As far as we are aware, that is the only proof of rapid mixing of the switch chain for a nontrivial class of graphs. Thus, we consider here extending the proof technique of~\cite{DyJeMu17} to a much larger class of graphs, which are not necessarily bipartite. To define this class, we need the following definition.
\subsection{Quasiclasses}\label{sec:quasigraph}
Let $\cC \subseteq \bip$, where $\bip$ denote the class of bipartite graphs.
Then we will define the class $\qua\cC$ as follows. A graph $G$ is in $\qua\cC$ if $G[\LR]\in \cC$ for all bipartitions $L,R$ of $V$. This may seem a very demanding definition, but it is not so for most classes of interest, as we shall see.

\begin{lemma}\label{lem:quasiclass1}
  If $\cC \subseteq \bip$ is a hereditary class, then so is $\qua\cC$.
\end{lemma}
\begin{proof}
  Suppose $G=(V,E)\in\qua\cC$ and $v\in V$. We wish to show that $G[V\setminus v]\in\qua\cC$. Let $L,R$ be any bipartition of $V\setminus v$. Then $L,R$ can be extended to a bipartition of $L\cup v,R$ of $V$. Thus $G[L\cup v{:}R]\in \cC$ and, since $\cC$ is hereditary, $G[\LR]\in\cC$. Thus $G[V\setminus v]\in\qua\cC$.
\end{proof}
\begin{lemma}\label{lem:quasiclass2}
  Let $\cC \subseteq \bip$ be a graph class that is hereditary and closed under
  disjoint union, then $\cC = \bip \cap \qua\cC$.
\end{lemma}
\begin{proof}
  First let $G = (L \cup R,E)$ be any bipartite graph that does not belong
  to $\cC$. Since $G = G[\LR]$, $G$ does not belong to $\qua\cC$.
  Hence $\cC \supseteq \bip \cap \qua\cC$.

  Next we show $\cC \subseteq \bip \cap \qua\cC$. Let $G = (X \cup Y,E)$
  be a graph in $\cC$ and let $\LR$ be any bipartition of $X \cup Y$.
  Now $G[\LR]$ is the disjoint union of $G_1 = G[(X \cap L)\cup(Y \cap R)]$
  and $G_2 = G[(X \cap R)\cup(Y \cap L)]$. The graphs $G_1$ and $G_2$ belong
  to $\cC$ since the class is hereditary, and hence $G[\LR]\in \cC$,  because
  $\cC$ is closed under disjoint union. Thus $G \in \qua\cC$.
\end{proof}
We also have
\begin{lemma}\label{lem:quasiclass3}
  $\qua\cbg=\och$.
\end{lemma}
\begin{proof}
  $G\notin\och$ if it has as even cycle $C$ with only even chords. Then $C$ is a hole in $G[\LR]$, for any bipartition $L,R$ of $V$ which is alternating on $C$. Thus $G[\LR]\notin\cbg$, so $G\notin\qua\cbg$. Conversely, suppose that $G\notin\qua$\cbg. Then there is some bipartition $L,R$ of $V$ such that $G[\LR]$ contains a hole $C$. The edges of $G[C]$ that are not in $G[\LR]$ must be even chords of $C$, so $C$ has only even chords in $G$. Thus $G\notin\och$.
\end{proof}
In~\cite{DyeMue18}, some other examples of \qua\ classes are discussed. As a final example here, in section \ref{sss:qchains} we consider the \qua class corresponding to the class \chain, of \emph{chain graphs}.

Our motivation for introducing this concept is that methods and results for bipartite graph classes may be easily extendible to the corresponding \qua class. In particular, we are interested in the case of \emph{monotone graphs}.

\subsection{Quasimonotone graphs}
For the class \mono, of monotone graphs, we will denote the hereditary (by Lemma~\ref{lem:quasiclass1}) class $\qua\mono$ by $\qmon$, and a graph $G\in\qmon$ will be called \emph{quasimonotone}. All monotone graphs are quasimonotone, by Lemma~\ref{lem:quasiclass2}.
Since \mono$\,\subset\,$\cbg, \qmon\,$\subset$\,\och, by Lemmas~\ref{lem:quasiclass1} and~\ref{lem:quasiclass3}. So the switch chain is ergodic on quasimonotone graphs, since we have \och\,$\subset$\,\switch.
\subsubsection{Unit interval graphs}
A \emph{unit interval graph} $G$ (also called a \emph{proper interval graph}, \emph{claw-free interval graph} or \emph{indifference graph}) is the \emph{intersection graph} of a set of unit intervals $v_i=[x_i,x_i+1]$ ($i\in[n]$)
on the real line. That is, $G=(V,E)$, where $V=\{v_i:i\in[n]\}$ and $v_iv_j\in E$ if and only if $i\neq j$ and $v_i\cap v_j\neq \emptyset$. The class of unit interval graphs will be denoted by \unii.

\unii\ is a hereditary class, with the following forbidden subgraphs: all chordless cycles $C_k$ of length $k\geq 4$, the \emph{claw}, the \emph{3-sun} and its complement, the \emph{net}. Our interest in this class results from the following.
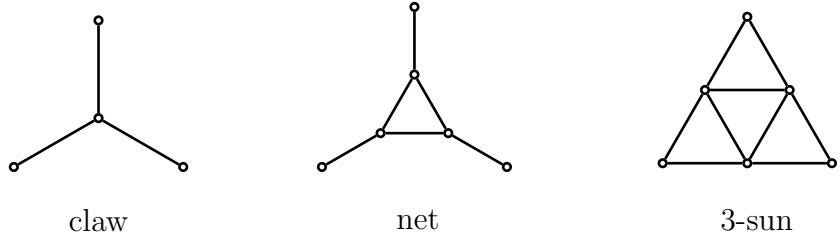
\begin{figure}[htb]
\centerline{\begin{tikzpicture}
\begin{scope}[yshift=0.6cm,scale=1.3]
\draw (0,0) node[vertex] (0) {} +(90:1)  node[vertex] (1) {}
+(-30:1)  node[vertex] (2) {} +(210:1)  node[vertex] (3) {} ;
\draw (0)--(1) (0)--(2) (0)--(3);
\end{scope}
\draw (0,-0.75) node {claw} ;
\begin{scope}[xshift=3.75cm,yshift=0.4cm,scale=0.9]
\draw (0,0) node[vertex] (0) {} +(60:1)  node[vertex] (1) {} +(0:1)  node[vertex] (2) {}
(0)+(210:1) node[vertex] (3) {} (2)+(-30:1) node[vertex] (4) {} (1)+(90:1) node[vertex] (5) {};
\draw (0)--(1)--(2)--(0)--(3) (2)--(4) (1)--(5);
\end{scope}
\draw (4.25,-0.75) node {net} ;
\begin{scope}[xshift=7.5cm,scale=0.75]
\draw (0,0) node[vertex] (0) {} +(60:1.5)  node[vertex] (1) {} +(60:3)  node[vertex] (2) {}
(1.5,0) node[vertex] (3) {} (3,0) node[vertex] (4) {} +(120:1.5)  node[vertex] (5) {};
\draw (0)--(3)--(4)--(5)--(2)--(1)--(0) (1)--(3)--(5)--(1) ;
\end{scope}
\draw (8.75,-0.75) node[empty] {3-sun} ;
\end{tikzpicture}}
\caption{Forbidden subgraphs for unit interval graphs}\label{quasi:fig20}
\end{figure}

\begin{theorem}\label{quasi:thm010}
  \unii\,$\subset$\,\qmon.
\end{theorem}
\begin{proof}
Let $G=(V,E)\in\unii$, and suppose that $L,R$ is any bipartition of $V$. Then, by definition, $G[\LR]$ is a \emph{unit interval bigraph}~\cite{HelHua04}. It is shown in~\cite{HelHua04} that the class of unit interval bigraphs coincides with the class \mono. Thus $G[\LR]$ is a monotone graph, and hence $G$ is quasimonotone.
\end{proof}
%\subsubsection*{Finding an initial perfect matching}
%In order to initialise the switch chain, we must find a perfect matching. We will show that there is a trivial $O(n)$ algorithm for doing this in unit interval graphs.
%
%Let $G$ be a unit interval graph, as above, and suppose that the $v_i$ are ordered so that $x_1\leq x_2\leq \cdots\leq x_n$.
%\begin{lemma}
%  Either $n$ is even and $M=\{v_{2i-1}v_{2i}:i\in[n/2]\}$ is a perfect matching, or $G$ has no perfect matching.
%\end{lemma}
%\begin{proof}
%  Clearly $n$ must be even, so suppose $n=2r$. Then $M$ will have no perfect matching if and only if $(v_{2i+1},v_{2i})\notin E$, for some $i\in[r]$. This is true if and only if $v_{2i-1}\cap v_{2i}=\emptyset$, which is equivalent to $x_{2i}>x_{2i-1}+1$. Then $v_j\cap v_k=\emptyset$ for all $j\leq 2i-1$ and $k\geq 2i$, so $v_jv_k\notin E$, for all $1\leq j\leq 2i-1$ and $2i\leq k\leq 2r$. Thus $G$ has two disconnected subgraphs $G',\,G''$ with vertex sets $V'=\{v_1,\ldots,v_{2i-1}\}$ and $V''=\{v_{2i},\ldots,v_{2r}\}$. Since $V'$ and $V''$ have odd numbers of vertices, neither $G'$ nor $G''$ can have a perfect matching. Hence neither can $G$.
%\end{proof}

Clearly \mono\,$\cup$\,\unii$\,\subseteq\,$\qmon, but the class is considerably larger than this, and there seems to be no simple characterisation of all graphs in the class. In Fig.~\ref{fig:quasiexample}(a), we give an example of a quasimonotone graph which is not monotone (because it is nonbipartite) and not unit interval (because it is not chordal).  In Fig.~\ref{fig:quasiexample}(b), we give an example of a quasimonotone graph which is chordal (so not monotone) but not unit interval (because it contains claws). We show in section~\ref{sec:rapidmixing} below that the switch chain is rapidly mixing in the class \qmon. Therefore, the applicability of the switch chain requires a recognition problem for quasimonotone graphs. In particular, can we recognise a quasimonotone graph in polynomial time? Trivially, this problem is only in \coNP, by guessing a bipartition $L,R$, and using an algorithm for recognising monotone graphs~\cite{SpBrSt87} to show that $G[\LR]$ is not monotone. However, we show in~\cite{DyeMue18} that the problem of quasimonotone graph recognition is in \Ptime.
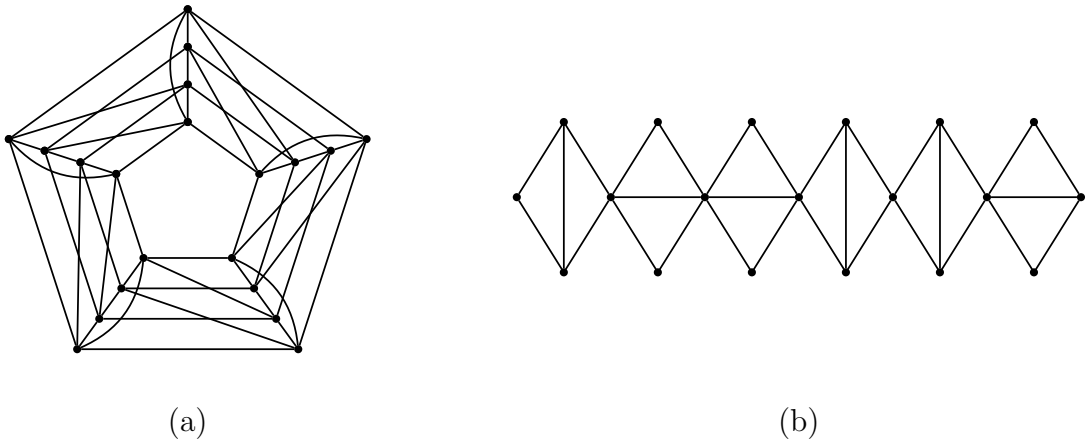
\begin{figure}[ht]
\begin{center}
  \tikzset{every node/.style={circle,draw,thick,fill=black,inner sep=0.8pt}}
  \begin{tikzpicture}[scale=0.5,line width=0.67pt]
    \foreach[count=\r from 2] \n in {11,12,13,14} \node (\n) at ( 90:\r) {};
    \foreach[count=\r from 2] \n in {15,16,17,18} \node (\n) at (162:\r) {};
    \foreach[count=\r from 2] \n in {19,20,21,22} \node (\n) at (234:\r) {};
    \foreach[count=\r from 2] \n in {23,24,25,26} \node (\n) at (306:\r) {};
    \foreach[count=\r from 2] \n in {27,28,29,30} \node (\n) at ( 18:\r) {};
    \foreach[count=\r from 2] \n in {31,32,33,34} \node (\n) at ( 90:\r) {};
    \foreach[count=\r from 2] \n in {35,36,37,38} \node (\n) at (162:\r) {};
    \foreach[count=\m from 15] \n in {11,12,...,30} \draw (\n)--(\m);
    \draw (11)--(12)--(13)--(14) to[bend right] (11)
          (15)--(16)--(17)--(18) to[bend right] (15)
          (19)--(20)--(21)--(22) to[bend right] (19)
          (23)--(24)--(25)--(26) to[bend right] (23)
          (27)--(28)--(29)--(30) to[bend right] (27);
    \draw (11)--(17)  (12)--(18)
          (15)--(21)  (16)--(22)
          (19)--(25)  (20)--(26)
          (23)--(29)  (24)--(30)
          (27)--(33)  (28)--(34);
    \draw (0,-6) node[empty] {(a)} ;
  \begin{scope}[xscale=1.25,xshift=6cm]
    \foreach \n in {2,4,6,8,10,12} \node (b\n) at (\n,-2) {} ;
    \foreach \n in {2,4,6,8,10,12} \node (t\n) at (\n,2) {} ;
    \foreach \n in {1,3,5,7,9,11,13} \node (m\n) at (\n,0) {} ;
    \foreach \n in {1,3,5,7,9,11} {\draw (m\n) -- ++(1,2) ;};
    \foreach \n in {3,5,7,9,11,13} {\draw (m\n) -- ++(-1,2);};
    \foreach \n in {1,3,5,7,9,11} {\draw (m\n) -- ++(1,-2) ;};
    \foreach \n in {3,5,7,9,11,13} {\draw (m\n) -- ++(-1,-2);};
    \draw (t2)--(b2) (t8)--(b8) (m3)--(m5)--(m7) (t10)--(b10) (m11)--(m13) ;
    \draw (7,-6) node[empty] {(b)} ;
  \end{scope}
  \end{tikzpicture}
\end{center}
\caption{Two quasimonotone graphs}\label{fig:quasiexample}
\end{figure}

\subsection{Rapid mixing of the switch chain}\label{sec:rapidmixing}

We will now show that the switch chain has polynomial time convergence on the class of quasimonotone graphs.

To do this, we simply extend to quasimonotone graphs the analysis for monotone graphs given in~\cite[Sec.\,3]{DyJeMu17}. We construct a canonical path between any pair of perfect matchings $X$, $Y$ in $G$ by considering the set of alternating cycles in $X\oplus Y$. Since quasimonotone graphs form a hereditary class, we can reduce the problem to constructing a canonical path for switching each of these cycles, taken in some canonical order. Each such cycle $H$ is an alternating Hamilton cycle in the graph $G'=G[H]$. Note that $G'$ is quasimonotone, by heredity, and has an even number $n$ of vertices, since $H$ is alternating. We will denote the restrictions of $X$ and $Y$ to $G'$ by $X'$ and $Y'$.

Now consider the alternating bipartition $L,R$ of $H$, which gives a bipartition of $G'$ such that
$|L|=|R|=n/2$. Since $G'$ is quasimonotone, $G'[\LR]$ is monotone, and we have $H\subseteq G'[\LR]$. Hence we can use the ``mountain climbing'' technique of~\cite{DyJeMu17} to construct a canonical path and an encoding for switching $X'$ to $Y'$ in $G'[\LR]$. This is also a canonical path for switching $X'$ to $Y'$ in $G'$, with length $O(n^2)$, as in~\cite{DyJeMu17}.

The rest of the analysis follows closely that in~\cite[Sec.\,3]{DyJeMu17}, noting only that $L,R$ each have at most $n/2$ vertices, rather than $n$, as in \cite{DyJeMu17}. However the conclusion, that the mixing time is $O(n^7\log n)$, remains the same. See~\cite{DyJeMu17} for further details.

Of course, the starting configuration for the switch chain must be a perfect matching. In the case of monotone graphs, a simple linear time algorithm was given in~\cite{DyJeMu17}. This does not extend to quasimonotone graphs, but the $O(n^3)$ algorithm of~\cite{Lawler76} for general graphs suffices to obtain $O(n^7\log n)$ mixing time. However, we know that an $O(n^2)$ algorithm exists, by making use of the quasimonotone structure. We will not give details here, since this is not a critical issue. We leave open the question of the existence of a $o(n^2)$ algorithm for quasimonotone graphs.

\subsubsection{Forbidden subgraphs of quasimonotone graphs}\label{sec:quasiforbid}

The forbidden subgraphs for the class \mono\ are  all (even) holes, together with the three 7-vertex graphs shown in Fig.~\ref{quasi:fig10}, as shown in~\cite{Kohl99}.
\begin{figure}[H]
\centerline{\begin{tikzpicture}[every node/.style={circle,draw,inner sep=0pt,minimum size=1.25mm},scale=0.8]
  \begin{scope}
    \path (0:0) node[fill=black] (c) {} (-30:1) node (a1) {} (-30:2) node[fill=black] (a2) {}
    (90:1) node (b1) {} (90:2) node[fill=black] (b2) {}  (210:1) node (c1) {} (210:2) node[fill=black] (c2) {};
     \draw (a2)--(a1) (b1)--(c)--(a1) (c)--(c1) (b2)--(b1) (c2)--(c1) ;
  \end{scope}
  \draw (0,-1.75) node[empty]  {tripod} ;
\begin{scope}[xshift=4.4cm,yshift=-1cm,xscale=0.9]
    \path  (2,0) node[fill=black] (a) {} (0,0) node (b) {} (2,0.75) node (c) {} (0,0.75) node[fill=black] (d) {}
     (2,1.5) node[fill=black] (e) {} (0,1.5) node (f) {}  (0,3) node[fill=black] (g) {} ;
     \draw  (g)--(f)--(e)--(c)--(d)--(b) (a)--(c) (d)--(f) ;
  \end{scope}
  \draw (5.25,-1.75) node[empty]  {armchair} ;
\begin{scope}[xshift=9cm,yshift=-1cm,xscale=0.75]
    \path (0,0) node (a) {} (1,0) node[fill=black] (b) {} (2,0) node (c) {}
    (0,1) node[fill=black] (d) {} (1,1) node (e) {}  (2,1) node[fill=black] (f) {} (1,3) node[fill=black] (g) {};
     \draw (a)--(b)--(c)--(f)--(e)--(d)--(a) (b)--(e)--(g) ;
  \end{scope}
  \draw (9.7,-1.75) node[empty]  {stirrer} ;
\end{tikzpicture}}
\caption{Forbidden subgraphs for monotone graphs}\label{quasi:fig10}
\end{figure}
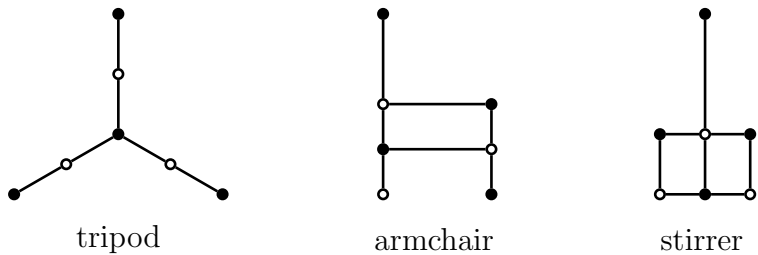
If $H$ is a bipartite graph, a graph $H'$ will be called \emph{pre-}$H$ if it has a bipartition $L,R$ such that $H'[\LR]\cong H$. Thus, if a class $\cC$ of bipartite graphs can be characterised by the set $\cF$ of forbidden subgraphs then $\qua\cC$ can characterised by forbidding all pre-$F$ ($F \in \cF$) as induced subgraphs.

We will call any pretripod, prearmchair or prestirrer, a \emph{flaw}. A \emph{flawless} graph $G$ will be one which contains no flaw as an induced subgraph. Let us call this (hereditary) class \noflaw. Since all flaws have only seven vertices, we can test in $O(n^7)$ time whether an input graph $G$ on $n$ vertices is flawless. Thus membership in \noflaw\ is certainly in \Ptime.

However, preholes can have unbounded size. It is easy to see that the preholes are all even cycles that have no odd chord, which is an infinite class. These preholes are clearly the forbidden subgraphs for the class \och. Thus quasimonotone graphs are characterised by the absence of preholes, pretripods, prestirrers and prearmchairs, which is equivalent to the statement $\qmon=\noflaw\,\cap\,\och$.

Unfortunately, this characterisation of \qmon\ does not seem to lead to polynomial time recognition. We have observed above that we do not know whether the class \och\ can be recognised in polynomial time, so we cannot simply test whether $G$ is flawless and odd chordal. However, we show in~\cite{DyeMue18} that quasimonotone graphs can be recognised in polynomial time.

\subsection{Quasi-chain graphs}\label{sss:qchains}

Chain graphs form a subclass of monotone graphs, and there is a trivial algorithm (see~\cite{DyJeMu17}) for counting matchings in such graphs.  However, this does not extend to the quasi-class. We know of no better analysis of the switch chain, and no better algorithm for either approximately or exactly counting matchings, than those we have given for quasimonotone graphs. However, we will show that there is a simpler recognition algorithm for graphs in this quasi-class than that given in ~\cite{DyeMue18}.

This class also illustrates a definitional issue. Thee class of chain graphs, \chain, is not closed under disjoint union, so the quasi-class does not include the class itself.
For example, in the simple chain graph of Fig.~\ref{chain:fig010}, $G[\LR]$ is a union of two chain graphs, so \qua\chain\  does not contain \chain. Therefore, we define instead the class \chains, which is the class of graphs such that every connected component is a chain graph.

\begin{figure}[ht]
\centering{%
\begin{tikzpicture}[xscale=0.7,yscale=1.2,font=\small]
\path
(0,0) node[b] (0) {}
(0,1) node[b] (1) {}
(2,0) node[b] (0') {}
(2,1) node[b] (1') {};
\draw (0)--(0')  (0)--(1')  (1)--(1')   ;
\draw (1,-0.75) node{$G$};
\end{tikzpicture}
\hspace*{2cm}
\begin{tikzpicture}[xscale=0.7,yscale=1.2,font=\small]
\path
(0,0) node[w] (0) {}
(0,1) node[b] (1) {}
(2,0) node[b] (0') {}
(2,1) node[w] (1') {};
\draw (0)--(0')  (0)--(1')  (1)--(1')   ;
\draw (1,-0.75) node{$L,R$};
\end{tikzpicture}
\hspace*{2cm}
\begin{tikzpicture}[xscale=0.7,yscale=1.2,font=\small]
\path
(0,0) node[w] (0) {}
(0,1) node[b] (1) {}
(2,0) node[b] (0') {}
(2,1) node[w] (1') {};
\draw (0)--(0')  (1)--(1')   ;
\draw (1,-0.75) node{$G[\LR]$};
\end{tikzpicture}}
\caption{}\label{chain:fig010}
\end{figure}

Now \qua\chains\ will contain \chains, by Lemma~\ref{lem:quasiclass2}. Clearly $\chains\subseteq\mono$, since each chain graph is monotone, and \mono\ is closed under disjoint union. However, it is known that the forbidden subgraph for the class \chain\ is $2K_2$, but we see from Fig.~\ref{chain:fig010} that $2K_2\in \chains$. So we need a different characterisation of \chains, in which the forbidden subgraphs are connected. This is not difficult. The following result seems to be folklore. It is stated without proof in~\cite[Proposition~2]{BraMos03} and~\cite[Property~1]{GeHeSc03}. (In fact~\cite[Proposition~2]{BraMos03} does have a two-line proof, but it is a proof of Corollary~\ref{cor:forb} below.)  It seems to be attributed to~\cite{BacTuz90}, which does contain related results, but not this. Therefore, we will give a short proof via monotone graphs. The graphs $P_5$ and $C_5$ are shown in Fig.~\ref{chain:fig030}.

\begin{lemma}\label{lem:P5}
  $G\in\chains$ if and only if it is bipartite and $P_5$-free.
\end{lemma}
\begin{proof}
  If $G$ is not bipartite, it is  clearly not in \chains. If it has an induced $P_5$, this must be entirely in some component $G'$ of $G$. But $P_5$ contains an induced $2K_2$, by deleting its middle vertex. So $G'$ cannot be a chain graph, and hence $G\notin\chains$.
  Conversely, suppose $G$ is bipartite and $P_5$-free. It cannot contain a flaw, or a $k$-hole for $k>4$, since the flaws contain an induced $P_5$, and so does every $k$-hole for any $k\geq 6$. See Fig.~\ref{chain:fig020}. Thus $G\in\mono$.

  \begin{figure}[hbt]
\centering{%
\begin{tikzpicture}[scale=0.75]
  \begin{scope}
    \path (0:0) node[b] (c) {} (-30:1) node[b](a1) {} (-30:2) node[b] (a2) {}
    (90:1) node[b](b1) {} (90:2) node[b] (b2) {}  (210:1) node[w](c1) {} (210:2) node[w] (c2) {};
     \draw  (b2)--(b1)--(c)--(a1)--(a2); \draw[dashed] (c)--(c1)--(c2);
  \end{scope}
  \draw (0,-1.75) node[empty]  {tripod} ;
\begin{scope}[xshift=5cm,yshift=-1cm,xscale=0.75]
    \path  (2,0) node[b] (a) {} (0,0) node[w](b) {} (2,0.75) node[b](c) {} (0,0.75) node[w] (d) {}
     (2,1.5) node[b] (e) {} (0,1.5) node[b](f) {}  (0,3) node[b] (g) {} ;
     \draw  (g)--(f)--(e)--(c) (a)--(c); \draw[dashed]  (b)--(d)--(f) (d)--(c) ;
  \end{scope}
  \draw (5.8,-1.75) node[empty]  {armchair} ;
\begin{scope}[xshift=10cm,yshift=-1cm,xscale=1]
    \path (0,0) node[b](a) {} (1,0) node[w] (b) {} (2,0) node[b](c) {}
    (0,1) node[b] (d) {} (1,1) node[b](e) {}  (2,1) node[b] (f) {} (1,3) node[w] (g) {};
     \draw (c)--(f)--(e)--(d)--(a); \draw[dashed] (a)--(b)--(c) (b)--(e)--(g) ;
  \end{scope}
  \draw (11.1,-1.75) node[empty]  {stirrer} ;
\begin{scope}[xshift=15cm,yshift=-1cm,xscale=0.5,yscale=0.67]
    \path (0,0) node[b](a) {} (0,2) node[b] (b) {} (2,2) node[b](c) {}
    (4,2) node[b] (d) {} (4,0) node[b](e) {}  (2,0) node[w] (f) {} ;
     \draw (a)--(b)--(c)--(d)--(e); \draw[dashed] (e)--(f)--(a) ;
  \end{scope}
  \draw (16.1,-1.75) node[empty]  {6-hole} ;
\end{tikzpicture}}
\caption{}\label{chain:fig020}
\end{figure}

Now consider any connected component $G'$ of $G$, with monotone biadjacency matrix $A'$.
If $G'$ is not a chain graph, then $A'$ must contain a submatrix of the form below, or its transpose.
\[\begin{array}{c|ccc}
&x&y\\ \hline
u&1&0\\
v&1&1\\
w&0&1
\end{array}\]
But this corresponds to an induced path $(u,x,v,y,w)$, which is a $P_5$, giving a contradiction. Thus $G'$ must be a chain graph, and $G\in\chains$.
\end{proof}
\begin{corollary}\label{cor:forb}
$G\in\chains$ if and only if it is (triangle,\thsp$C_5,P_5$)-free.
\end{corollary}
\begin{proof}
  If $G$ is $P_5$-free, it has no holes of size 6 or more. Therefore, unless it has a triangle or a 5-hole, it must be bipartite. So, if we exclude these two possibilities, Lemma~\ref{lem:P5} implies $G\in\chains$. The converse is also clear from Lemma~\ref{lem:P5}. All graphs in \chains\ are $P_5$-free and bipartite, so cannot have a triangle or a 5-cycle.
\end{proof}
\begin{corollary}
\qua\chains\ is precisely the class of (pre-$P_5$)-free graphs, and membership can be recognised in $O(n^5)$ time.
\end{corollary}
\begin{proof}
  From Lemma~\ref{lem:P5}, the forbidden subgraphs for \qua\chains\ are pre-$P_5$'s and preholes. But any prehole has size 6 or more, and so induces a pre-$P_5$. Thus preholes give no new forbidden subgraphs. The pre-$P_5$'s have only 5 vertices, so they can be searched for by brute force in $O(n^5)$ time.
\end{proof}
The class pre-$P_5$ contains the ten graphs shown in Fig.~\ref{chain:fig030}, up to isomorphism. They are in order of the number of edges added to $P_5$, which is given an alternating bipartition.
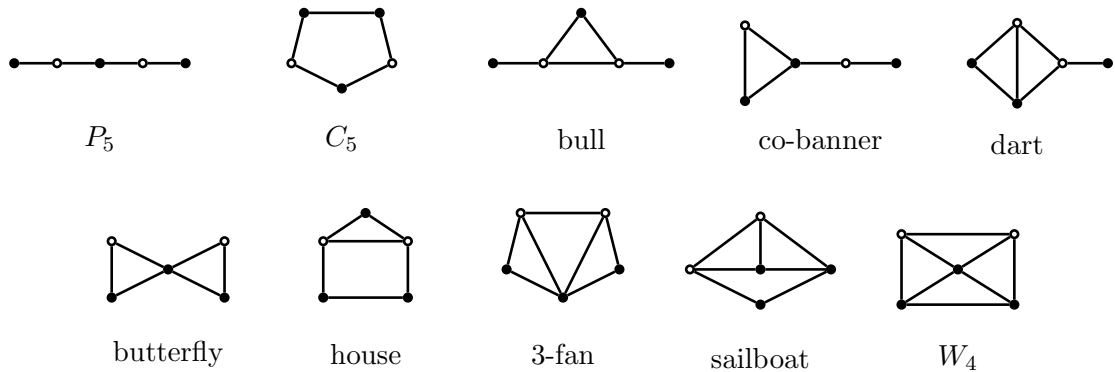
\begin{figure}[htb]
\centering{
\begin{tikzpicture}[scale=0.67,font=\small]
\begin{scope}[xscale=0.85]
\draw (0,0) node[b] (0) {} (1,0)  node[w] (1) {}
(2,0)  node[b] (2) {} (3,0)  node[w] (3) {}  (4,0)  node[b] (4) {} ;
\draw (0)--(1)--(2)--(3)--(4) ;
\draw (2,-1.5) node {$P_5$} ;
\end{scope}
\begin{scope}[xshift=5.5cm,yscale=1]
\draw (0,0) node[w] (0) {} (0.25,1)  node[b] (1) {}
(1,-0.5)  node[b] (2) {} (1.75,1)  node[b] (3) {}  (2,0)  node[w] (4) {} ;
\draw (2)--(0)--(1) (3)--(4)--(2) (1)--(3);
\draw (1,-1.5) node {$C_5$} ;
\end{scope}
\begin{scope}[xshift=9.5cm]
\draw (0,0) node[b] (0) {} (1,0)  node[w] (1) {}
(2.5,0)  node[w] (2) {} (3.5,0)  node[b] (3) {}  (1.75,1)  node[b] (4) {} ;
\draw (0)--(1)--(2)--(3) (1)--(4)--(2) ;
\draw (1.75,-1.5) node {bull} ;
\end{scope}
\begin{scope}[xshift=14.5cm,yscale=1.5]
\draw (0,-0.5) node[b] (0) {} (0,0.5)  node[w] (1) {}
(1,0)  node[b] (2) {} (2,0)  node[w] (3) {}  (3,0)  node[b] (4) {} ;
\draw (2)--(0)--(1)--(2)--(3)--(4) ;
\draw (1.5,-1) node {co-banner} ;
\end{scope}
\begin{scope}[xshift=19cm,yscale=0.8,xscale=0.9]
\draw (0,0) node[b] (0) {} (1,-1)  node[b] (1) {}
(2,0)  node[w] (2) {} (3,0)  node[b] (3) {}  (1,1)  node[w] (4) {} ;
\draw (0)--(1)--(2)--(3) (1)--(4)--(2) (0)--(4) ;
\draw (1,-2) node {dart} ;
\end{scope}
\end{tikzpicture}\\[3ex]
\begin{tikzpicture}[scale=0.75,font=\small]
\begin{scope}
\draw (0,-0.5) node[b] (0) {} (0,0.5)  node[w] (1) {}
(1,0)  node[b] (2) {} (2,0.5)  node[w] (3) {}  (2,-0.5)  node[b] (4) {} ;
\draw (2)--(0)--(1)--(2)--(3)--(4)--(2) ;
\draw (1,-1.5) node {butterfly} ;
\end{scope}
\begin{scope}[xshift=3.75cm,yshift=-0.5cm,xscale=0.75,yscale=1]
\draw (0,0) node[b] (0) {} (0,1)  node[w] (1) {}
(2,1)  node[w] (2) {} (2,0)  node[b] (3) {}  (1,1.5)  node[b] (4) {} ;
\draw (0)--(1)--(2)--(3)--(0) (1)--(4)--(2);
\draw (1,-1) node {house} ;
\end{scope}
\begin{scope}[xshift=7cm,yscale=1]
\draw (0,0) node[b] (0) {} (0.25,1)  node[w] (1) {}
(1,-0.5)  node[b] (2) {} (1.75,1)  node[w] (3) {}  (2,0)  node[b] (4) {} ;
\draw (2)--(0)--(1)--(2)--(3)--(4)--(2) (1)--(3);
\draw (1,-1.5) node {3-fan} ;
\end{scope}
\begin{scope}[xshift=10.25cm,scale=1.25]
\draw (0,0) node[w] (0) {} (1,-0.5)  node[b] (1) {}
(2,0)  node[b] (2) {} (1,0)  node[b] (3) {}  (1,0.75)  node[w] (4) {} ;
\draw (0)--(1)--(2)--(4)--(3) (0)--(4) (0)--(3)--(2);
\draw (1,-1.25) node {sailboat} ;
\end{scope}
\begin{scope}[xshift=14cm,yscale=1.25]
\draw (0,-0.5) node[b] (0) {} (0,0.5)  node[w] (1) {}
(1,0)  node[b] (2) {} (2,0.5)  node[w] (3) {}  (2,-0.5)  node[b] (4) {} ;
\draw (2)--(0)--(1)--(2)--(3)--(4)--(2) (1)--(3) (0)--(4);
\draw (1,-1.25) node {$W_4$} ;
\end{scope}
\end{tikzpicture}}
\caption{Forbidden subgraphs for \qua\chains}\label{chain:fig030}
\end{figure}

\begin{lemma}\label{lem:holefree}
  Let $\cC$ be a hereditary bipartite graph class. Then $\qua\cC\subseteq\textsc{HoleFree}$ if and only if $\cC\subseteq\chains$.
\end{lemma}
\begin{proof}
If $G\notin\qua\textsc{HoleFree}$, then $G$ contains a hole $H$. Then any bipartition of $G$ which extends an alternating bipartition of $H$ contains $P_5$ as a subgraph, Thus $G\notin\qua\textsc{Chains}$, by Lemma~\ref{lem:P5}. Thus $\qua\textsc{Chains}\subseteq\textsc{HoleFree}$.

Now suppose $G\notin\textsc{Chains}$. Then $G$ contains a $P_5$ by Lemma~\ref{lem:P5}. Thus, by heredity, $\cC$ contains $P_5$ and all its subgraphs. Then $\qua\cC$ contains $C_5$, since every bipartition of $C_5$ gives $P_5$ or its subgraphs. Thus $\qua\cC\nsubseteq\textsc{HoleFree}$.
\end{proof}

In particular, we see that if $\cC\nsubseteq\textsc{Chains}$, then $\qua\cC\nsubseteq\textsc{Perfect}$, by \cite{ChRoST06}.

\section{Slow mixing of the switch chain}\label{sec:slowmixing}

Unfortunately, the switch chain appears to mix slowly in the worst case on graphs in many hereditary classes of interest. In this section we consider the two classes \intl\ and \perm, by showing that even their intersection \chp\ exhibits slow mixing.

\subsection{Chordal permutation graphs}

The examples we present here are an inspired by those given for biconvex graphs in \cite{Blumbe12,Matthe08}.

\subsubsection{Construction}

For every integer $k \ge 1$ let $G_k$ be the graph with vertex set
$U \cup W \cup X \cup Y \cup Z$ and edge set
$\EUW \cup \EW \cup \EWX \cup \EX \cup \EXY \cup \EY \cup \EYZ$
defined by
\allowdisplaybreaks
\begin{align*}
  U &= \{u_i \mid 1 \le i \le k\} &
  Z &= \{z_i \mid 1 \le i \le k\} \\
  W &= \{w_i \mid 1 \le i \le k\} &
  Y &= \{y_i \mid 1 \le i \le k\} \\
  X &= \{x_i \mid 1 \le i \le 2\} &
  \EX  &= \{x_1x_2\}\\
  \EUW &= \{u_iw_j \mid 1 \le i \le j \le k\}  &
  \EYZ &= \{z_iy_j \mid 1 \le i \le j \le k\}  \\
  \EW  &= \{vw \mid w \in W, v \in W\sm\{w\}\} &
  \EY  &= \{vy \mid y \in Y, v \in Y\sm\{y\}\} \\
  \EWX &= \{vx \mid x \in X, v \in U \cup W\}  &
  \EXY &= \{vx \mid x \in X, v \in Y \cup Z\}  \\
\end{align*}
Thus $G_k$ has $n=4k+2$ vertices.

Using the notation for cographs, i.e.\ $\uplus$ for disjoint union and
$\Join$ for complete join, we have $G_k = (X,\EX) \Join \big((U \cup W,
\EUW\cup\EW)\uplus(Y \cup Z, \EYZ\cup\EY)\big)$. The graphs $G_k[U \cup W]$
and $G_k[Y \cup Z]$ are threshold graphs,
% which have exactly one perfect matching
that is, these graphs are both interval and permutation graphs, see
Fig.~\ref{fig:Hasse}. Since both these classes are closed under
disjoint union and join with complete graphs, $G_k$ too is both
an interval graph and a permutation graph.
For illustration, Fig.~\ref{fig:G4} gives $G_4$, where $w_2,w_3,y_2,y_3$ are  not labelled for clarity. Then Fig.~\ref{fig:G4int} gives an interval model of $G_4$, and Fig.~\ref{fig:G4perm} gives a permutation model.
\tikzset{vertex/.style={circle,draw,fill=none,inner sep=1.2pt}}
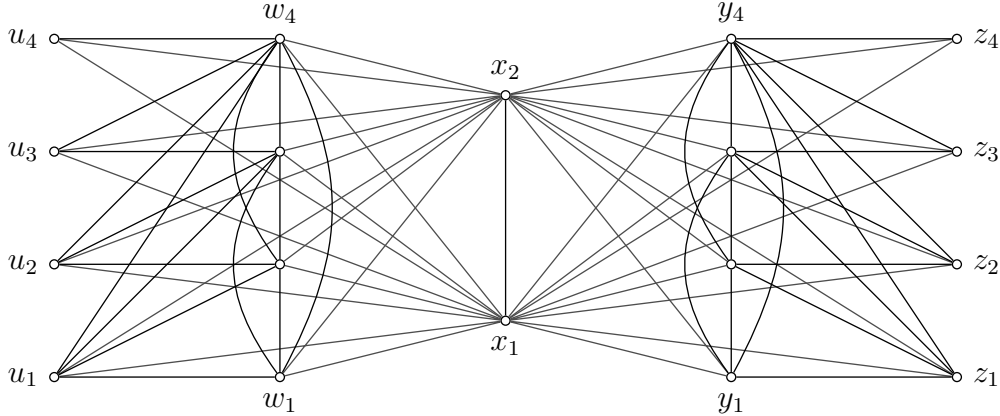
\begin{figure}[htbp]
\begin{center}
  \begin{tikzpicture}[xscale=3,yscale=1.5,line width=0.5pt]
    \node[vertex] at (0,1.5) (x1) [label=below:$x_1$] {};
    \node[vertex] at (0,3.5) (x2) [label=above:$x_2$] {};
    \foreach \ind in {1,2,3,4}
      \node[vertex] at (-2,\ind) (u\ind) [label=left:$u_{\ind}$] {};
    \foreach \ind in {1,2,3,4}
      \node[vertex] at (-1,\ind) (w\ind) {}; % $w_{\ind}$
    \foreach \ind in {1,2,3,4}
      \node[vertex] at (1,\ind) (y\ind) {}; % $y_{\ind}$
    \foreach \ind in {1,2,3,4}
      \node[vertex] at (2,\ind) (z\ind) [label=right:$z_{\ind}$] {};
    \draw (x1)--(x2);
    \foreach \ind in {1,2,3,4} \draw[color=black!70] (x1)--(u\ind)--(x2)
      (x1)--(w\ind)--(x2) (x1)--(y\ind)--(x2) (x1)--(z\ind)--(x2);
    \foreach \ind in {1,2,3,4} \foreach \jot in {1,...,\ind}
      \draw (u\jot)--(w\ind)  (z\jot)--(y\ind);
    \foreach \ind/\jot in {1/2,2/3,3/4}
      \draw (w\ind)--(w\jot)  (y\ind)--(y\jot);
    \draw (w1) to [bend left=20] (w3);
    \draw (w2) to [bend left=20] (w4);
    \draw (w1) to [bend right=15] (w4);
    \draw (y1) to [bend  left=20] (y3);
    \draw (y2) to [bend  left=20] (y4);
    \draw (y1) to [bend  right=15] (y4);
    \node [below] at (w1.south) {$w_1$};
    \node [above] at (w4.north) {$w_4$};
    \node [below] at (y1.south) {$y_1$};
    \node [above] at (y4.north) {$y_4$};
  \end{tikzpicture}
  \end{center}\vspace{-\baselineskip}
  \caption{The graph $G_4$}
  \label{fig:G4}
  \end{figure}

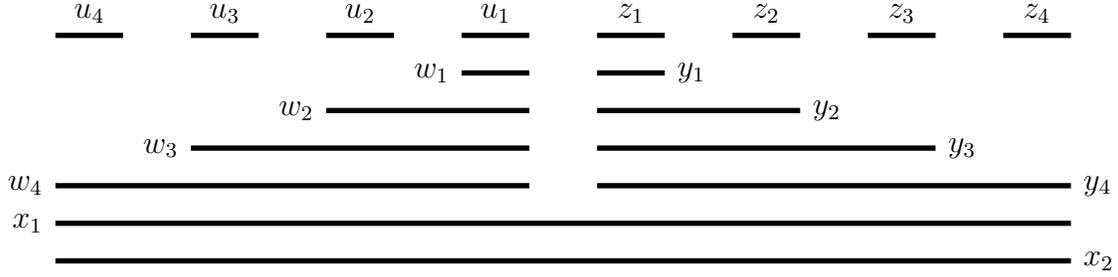
\begin{figure}[htbp]
  \begin{center}
  \begin{tikzpicture}[xscale=0.45, yscale=0.5, label distance=-4pt]
    \foreach[count=\ind from 1] \pos in {2,6,10,14}
      \node at (-\pos,6) [label=above:$u_{\ind}$] {};
    \foreach \le/\ri in {1/3, 5/7, 9/11, 13/15}
      \draw[line width=2pt] (-\le,6)--(-\ri,6);
    \foreach \hi/\ri in {5/3, 4/7, 3/11, 2/15}
      \draw[line width=2pt] (-1,\hi)--(-\ri,\hi);
    \foreach[count=\ind from 1] \hi/\ri in {5/3, 4/7, 3/11, 2/15}
      \node at (-\ri,\hi) [label=left:$w_{\ind}$] {};
    \foreach[count=\ind from 1] \pos in {2,6,10,14}
      \node at (\pos,6) [label=above:$z_{\ind}$] {};
    \foreach \le/\ri in {1/3, 5/7, 9/11, 13/15}
      \draw[line width=2pt] (\le,6)--(\ri,6);
    \foreach \hi/\ri in {5/3, 4/7, 3/11, 2/15}
      \draw[line width=2pt] (1,\hi)--(\ri,\hi);
    \foreach[count=\ind from 1] \hi/\ri in {5/3, 4/7, 3/11, 2/15}
      \node at (\ri,\hi) [label=right:$y_{\ind}$] {};
    \draw[line width=2pt] (-15,1)--(15,1);
    \node at (-15,1) [label=left:$x_1$] {};
    \draw[line width=2pt] (-15,0)--(15,0);
    \node at (15,0) [label=right:$x_2$] {};
  \end{tikzpicture}
\end{center}\vspace{-\baselineskip}
  \caption{An interval model of $G_4$}
  \label{fig:G4int}
\end{figure}

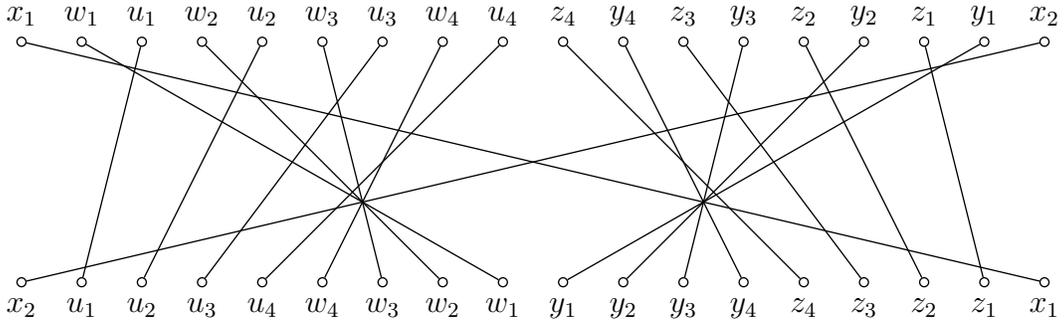
\begin{figure}[htbp]
\begin{center}
  \begin{tikzpicture}[xscale=0.4,yscale=0.4,line width=0.5pt]
    \foreach[count=\ind from 1] \x in {13, 9, 5, 1} \node[vertex] at (\x, 4)
      (zu\ind) [label=above:$z_{\ind}$] {};
    \foreach[count=\ind from 1] \x in {15,13,11, 9} \node[vertex] at (\x,-4)
      (zl\ind) [label=below:$z_{\ind}$] {};
    \foreach \ind in {1,2,3,4}  \draw (zu\ind)--(zl\ind);
    \foreach[count=\ind from 1] \x in {15,11, 7, 3} \node[vertex] at (\x, 4)
      (yu\ind) [label=above:$y_{\ind}$] {};
    \foreach[count=\ind from 1] \x in { 1, 3, 5, 7} \node[vertex] at (\x,-4)
      (yl\ind) [label=below:$y_{\ind}$] {};
    \foreach \ind in {1,2,3,4}  \draw (yu\ind)--(yl\ind);
    \foreach[count=\ind from 1] \x in {13, 9, 5, 1} \node[vertex] at (-\x, 4)
      (uu\ind) [label=above:$u_{\ind}$] {};
    \foreach[count=\ind from 1] \x in {15,13,11, 9} \node[vertex] at (-\x,-4)
      (ul\ind) [label=below:$u_{\ind}$] {};
    \foreach \ind in {1,2,3,4}  \draw (uu\ind)--(ul\ind);
    \foreach[count=\ind from 1] \x in {15,11, 7, 3} \node[vertex] at (-\x, 4)
      (wu\ind) [label=above:$w_{\ind}$] {};
    \foreach[count=\ind from 1] \x in { 1, 3, 5, 7} \node[vertex] at (-\x,-4)
      (wl\ind) [label=below:$w_{\ind}$] {};
    \foreach \ind in {1,2,3,4}  \draw (wu\ind)--(wl\ind);
    \node[vertex] at (-17, 4) (xu1) [label=above:$x_1$] {};
    \node[vertex] at ( 17,-4) (xl1) [label=below:$x_1$] {};
    \draw (xu1)--(xl1);
    \node[vertex] at ( 17, 4) (xu2) [label=above:$x_2$] {};
    \node[vertex] at (-17,-4) (xl2) [label=below:$x_2$] {};
    \draw (xu2)--(xl2);
  \end{tikzpicture}
\end{center}\vspace{-\baselineskip}
  \caption{A permutation model of $G_4$}
  \label{fig:G4perm}
\end{figure}

\subsubsection{Perfect matchings of \boldmath$G_k$}

Now we consider the perfect matchings of $G_k$. One of them is
\[ M_0 = \{u_iw_i, y_iz_i \mid i \in[k]\} \cup \{x_1x_2\} \]
and plays a special role. If a perfect matching $M$ of $G_k$
contains the edge $x_1x_2$ then $M=M_0$ because the threshold graphs
$G_k[U \cup W]$ and $G_k[Y \cup Z]$ have only one perfect matching,
namely $M_0 \cap \EUW$ and $M_0 \cap \EYZ$.

No perfect matching of $G_k$ matches one vertex in $X$ to a vertex in
$U \cup W$ and the other to a vertex in $Y \cup Z$, because, for every
$v_1 \in U \cup W$ and $v_2 \in Y \cup Z$, the graph $G_k\sm\{v_1,x_1,v_2,x_2\}$
contains two odd components. For $v_1 \ne w_k$ and $v_2 \ne y_k$ it
consists of two connected components that contain $2k-1$ vertices each.

That is, every perfect matching $M$ of $G_k$ either contains the edge
$x_1x_2$ or it contains edges $x_1v_1$ and $x_2v_2$ where either
$v_1,v_2 \in U \cup W$ or $v_1,v_2 \in Y \cup Z$. We call this the
one-sided property of the perfect matchings of $G_k$.

Let $\cM$ be the set of perfect matchings of $G_k$ and let
$\cM'_1$ and $\cM'_2$ be the set of perfect matchings of
$G_k[U \cup W \cup X]$ and $G_k[X \cup Y \cup Z]$, respectively.
With
\begin{align*}
% \cM_1 &= \{M \cup \{x_iy_i \mid 1 \le i \le k\} \mid M \in \cM'_1\} \\
  \cM_1 &= \{M \cup (M_0 \cap \EXY) \mid M \in \cM'_1\} \\
% \cM_2 &= \{M \cup \{u_iw_i \mid 1 \le i \le k\} \mid M \in \cM'_2\} \\
  \cM_2 &= \{M \cup (M_0 \cap \EUW) \mid M \in \cM'_2\}
\end{align*}
we have
\begin{align*}
  \cM     &= \cM_1 \cup \cM_2 &
  \{M_0\} &= \cM_1 \cap \cM_2
\end{align*}
by the one-sided property shown above. From $|\cM'_i| = 3^k$ for $i=1,2$
follows $|\cM| = 2\cdot3^k-1$.

\subsection{Mixing time}

Note that $\G{G_k}$ is connected, but $\G{G_k} \thsp\sm\thsp M_0$ is not.
By induction we show that every matching $M \in \cM$ is at most $k$
switches away from $M_0$. This is obvious for $M=M_0$. In the
inductive step we may assume $M \in \cM_1 \sm\thsp M_0$ by symmetry.
We consider the maximal index $i$ such that $u_iw_i \notin M$. Let
$u_ix$ and $w_iv$ be the two edges in $M$ that saturate $u_i$ and
$w_i$. Since $u_jw_j \in M$ for $i<j \le k$ we have $x \in X$ and $v
\in W \cup X$. Hence $vx$ is an edge of $G_k$. Switching the $4$-cycle
$(u_i,w_i,v,x)$ transforms $M$ into a matching containing the edges $u_jw_j$
for all indices $j$ with $i \le j \le k$. By induction, the distance
from $M$ to $M_0$ in $\cG_k$ is at most~$k$. We may also observe that,
$|\nbh{M_0}\cap\cM_1|=k$ in \G{G_k}, and $|\nbh{M_0}\cap\cM_2|=k$.

Now, similarly to~\cite{Blumbe12,Matthe08}, we upper bound the conductance of the switch
chain by computing the flow through the cut $\cM_1\setminus M_0:\cM_2$.
There are only $k$ edges in the cut, those from $M_0$ to $\cM_1$, and each has
transition probability $2/n^2$.
The uniform equilibrium distribution $\pi$ of the chain gives every state $M\in\cM$
probability $\pi(M)=1/|\cM|$, and thus $\pi(\cM_1\setminus M_0)<\nicefrac12$.
Thus the flow through the cut is at most $2k/(n^2|\cM|)<1/(8k|\cM|)$, and hence the conductance
of the chain is
\[ \Phi\ \leq\ \frac{2}{8k|\cM|}\ =\ \frac{1}{4k(2.3^k-1)}\,.  \]
Now, for example from~\cite[Thm.\,7.3]{LePeWi06}, the mixing time \tmix for the chain to reach variation distance~\nicefrac14 from $\pi$ satisfies $\tmix \geq 1/(4\Phi)$. Thus, for the switch chain on $G_n$,
\[\tmix\ \geq\ \frac{4k(2.3^k-1)}{4}\ =\ k(2.3^k-1)\ >\ 3^{k+1}\ > \ 3^{(n+2)/4}\,, \]
for all $k\geq 2$. Thus the mixing time of the switch chain increases exponentially  for the graph sequence $G_n$
$(n=10,14,18,\ldots)$.

\section{The switch chain in general graphs}\label{sec:generalswitch}

We have considered the ergodicity and rapid mixing of the switch chain in hereditary classes where all graphs have the property. However, we might ask about recognising the ergodicity, or rapid mixing, of the switch chain for an arbitrary graph. We have seen that there are graphs that are not ergodic, and ergodic graphs that are not rapidly mixing. So we might wish to establish the complexity of recognising ergodicity, or rapid mixing. Since recognising rapid mixing is at least as hard as recognising ergodicity, we will consider only the ergodicity question. Consider the following computational problems, where we measure the complexity of the problem as a function of the graph size~$n$.

\textbf{Ergodicity}\\
\textsf{Input}: A graph $G$ on $n$ vertices.\\
\textsf{Question}: Is $\G{G}$ a connected graph?

or the seemingly simpler problem,

\textbf{Connection}\\
\textsf{Input}: A graph $G$ on $n$ vertices, and two perfect matchings $X$, $Y$ in $G$.\\
\textsf{Question}: Are $X$, $Y$ in the same component of $\G{G}$?

\textbf{Connection} is easily seen to be in \Pspace. It is the $st$-connectivity problem on a graph with less than $n^n$ vertices and degrees less than $n^2$. Then, since $st$-connectivity is in \Lspace (log-space)~\cite{Reingo08}, it follows that \textbf{Connection} is in \Pspace. Thus \textbf{Ergodicity} is also in \Pspace.  We simply guess two matchings which are disconnected, and use the \textbf{Connection} algorithm to prove disconnection in \Pspace. So \textbf{Ergodicity} is in $\Pspace^{\NP} = \Pspace$.
It is possible that the problem is $\Pspace$-complete, but we have no evidence for this.

However, it is not clear that either problem is in \NP, or  even in the polynomial hierarchy, though we suspect that this is the case. We could place \textbf{Connection} in \NP if we had a polynomial bound on the diameter of \G{G}. Then, from the argument above, \textbf{Ergodicity} could be solved in \coNP using an oracle for \textbf{Connection}, which would place it within the first two levels of the polynomial hierarchy.

Thus we might first ask: what is the maximum diameter of \G{G}, over all ergodic graphs $G$ on $n$ vertices. in particular, is this polynomially bounded?

For hereditarily ergodic graphs, we showed, in Lemma~\ref{lem:diamswitch}, that \G{G} has diameter $O(n)$. However, this is not true in general. In the following, we show that the diameter of the switch chain can be $\Omega(n^2)$ for a graph on which it is ergodic. Of course, this gives a rather weak lower, rather than an upper, bound on the diameter. But it does show that there is not necessarily a ``monotonically improving'' path from a matching $X$ to a matching $Y$ in $G$. And the difficulty of proving even this weak result suggests that establishing a polynomial upper bound will be far from easy.

\subsection{The spider's web graph}
\tikzset{vertex/.style={circle,draw,fill=none,inner sep=1pt}}
Let $\jj$ denote $j\bmod6$. The \emph{spider's web} graph $W_k$ is $(V_k,E_k)$,
\begin{align*}
  V_0 &= \emptyset,\quad
  U_k  = \{u_{kj}:j\in[6]\},\quad
  V_k  = V_{k-1}\cup U_k\quad (k\geq 1)\,.\\
  E_1 &= \{(u_{11},u_{14})\}\cup\{(u_{1j}, u_{1,\,[j]+1}) :j\in[6]\}\,,\\
  E_k &= E_{k-1} \cup\{(u_{k-1,j},u_{kj}), (u_{kj}, u_{k,\,[j]+1}) :j\in[6]\},\quad(k\geq 1)\,.
\end{align*}
For example, $W_5$ is shown in Fig.~\ref{fig:spider}. Note that $W_k$ is bipartite, with bipartition $V_{k,0},\,V_{k,1}$, where $V_{k,p}=\{u_{ij}: i+j=p \bmod 2\}$. We will also define the following subgraphs: the \emph{hexagon}
$C_i=W_k[U_i]$ ($1\in[k]$), and the \emph{annulus} $A_i=W_k[U_i\cup U_{i+1}]$ ($i\in[k-1]$).
Clearly $C_i\simeq C_1$, for all $i\in[k]$, and $A_i\simeq A_1$, for all $i\in[k-1]$. Also $W_k[V_i]\simeq W_i$ for any $i\in[k]$, so we may refer to this subgraph simply as $W_i$.

\begin{figure}[htbp]
\begin{center}
  \begin{tikzpicture}[scale=0.45, line width=0.75pt]
    \foreach \angle in {0,60,120,180,240,300}
       \foreach \radius in {1,2,3,4,5}
          \node[vertex] at (\angle:\radius) (z\radius\angle) {};
    \foreach \angle in {0,60,120,180,240,300}
       \foreach[count=\ind from 1] \radius in {2,3,4,5}
          \draw (z\ind\angle)--(z\radius\angle);
    \foreach \a/\b in {0/60, 60/120, 120/180, 180/240, 240/300, 300/0}
       \foreach \radius in {1,2,3,4,5}
          \draw (z\radius\a)--(z\radius\b);
    \draw (z10)--(z1180);
    \node at (0:6) {$u_{51}$}; \node at (60:5.9) {$u_{52}$}; \node at (120:5.9) {$u_{53}$};
    \node at (180:6) {$u_{54}$}; \node at (240:5.9) {$u_{55}$}; \node at (300:5.9) {$u_{56}$};
  \end{tikzpicture} \hspace*{1cm}
\end{center}\vspace{-0.5\baselineskip}
\caption{Spider's web $W_5$}\label{fig:spider}
\end{figure}
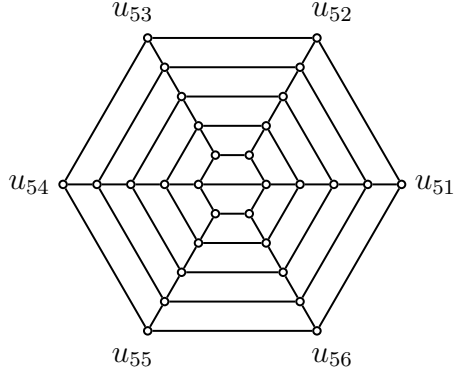

Note that $W_k$ is not hereditarily ergodic for any $k>1$. This follows from \cite[Lem.~2]{DyJeMu17}, but note that we have $A_1\subset W_k$, and the matchings $M_1$, $M_2$ in Fig.~\ref{fig:spiderNon-ergodic} have no switches in~$A_1$.

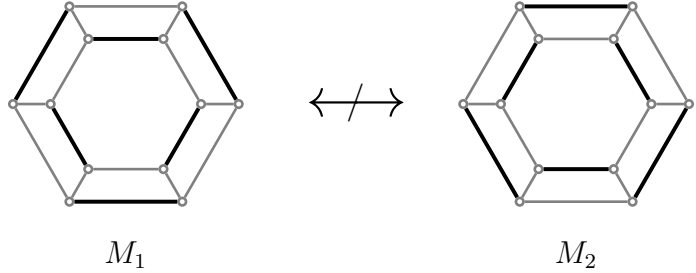
\begin{figure}[htbp]
\begin{center}
\begin{tikzpicture}[scale=0.5,line width=1pt,gray]
    \foreach \angle in {0,60,120,180,240,300}
       \foreach \radius in {2,3}
          \node[vertex] at (\angle:\radius) (z\radius\angle) {};
    \foreach \angle in {0,60,120,180,240,300}
       \foreach[count=\ind from 2] \radius in {3}
          \draw (z\ind\angle)--(z\radius\angle);
    \foreach \a/\b in {0/60, 120/180, 240/300}
       \foreach \radius in {2}
          \draw (z\radius\a)--(z\radius\b);
    \foreach \a/\b in {0/60, 120/180, 240/300}
       \foreach \radius in {3}
          \draw[black,line width=1.5pt] (z\radius\a)--(z\radius\b);
    \foreach \a/\b in {60/120, 180/240, 300/0}
       \foreach \radius in {2}
          \draw[black,line width=1.5pt] (z\radius\a)--(z\radius\b);
    \foreach \a/\b in {60/120, 180/240, 300/0}
       \foreach \radius in {3}
          \draw (z\radius\a)--(z\radius\b);
          \node[black] at (0,-4) {$M_1$};
  \end{tikzpicture}%
\qquad\raisebox{21.5mm}{\LARGE$\longleftrightarrow$}\hspace*{-8mm}\raisebox{22mm}{\large /}\qquad\quad
  \begin{tikzpicture}[scale=0.5, line width=1pt,gray]
    \foreach \angle in {0,60,120,180,240,300}
       \foreach \radius in {2,3}
          \node[vertex] at (\angle:\radius) (z\radius\angle) {};
    \foreach \angle in {0,60,120,180,240,300}
       \foreach[count=\ind from 2] \radius in {3}
          \draw (z\ind\angle)--(z\radius\angle);
    \foreach \a/\b in {0/60, 120/180, 240/300}
       \foreach \radius in {2}
          \draw[black,line width=1.5pt] (z\radius\a)--(z\radius\b);
    \foreach \a/\b in {60/120, 180/240, 300/0}
       \foreach \radius in {3}
          \draw[black,line width=1.5pt] (z\radius\a)--(z\radius\b);
    \foreach \a/\b in {60/120, 180/240, 300/0}
       \foreach \radius in {2}
          \draw (z\radius\a)--(z\radius\b);
    \foreach \a/\b in {0/60, 120/180, 240/300}
       \foreach \radius in {3}
          \draw (z\radius\a)--(z\radius\b);
          \node[black] at (0,-4) {$M_2$};
  \end{tikzpicture}
\end{center}\vspace{-\baselineskip}
\caption{Non-ergodicity of $A_1$}\label{fig:spiderNon-ergodic}
\end{figure}

However, these two matchings are the only obstructions to ergodicity.
\begin{lemma}\label{lem:ergodicA1}
\G{A_1} comprises a connected component, and two isolated vertices, $M_1$ and $M_2$.
\end{lemma}
\begin{proof}
  Let $M$ be any matching in $A_1$. Suppose first that $M$ has a \emph{cross} edge $u_{1j}u_{2j}$ for some $j\in[6]$. The graph $A_1\setminus\{u_{1j},\,u_{2j}\}$, given by deleting $u_{1j},\,u_{2j}$, is hereditarily ergodic, by~\cite[Lem.~2]{DyJeMu17}. Thus, $M$ is connected to the matching $M_0=\{u_{1j}u_{2j}:j\in[6]\}$ with all cross edges, as shown in Fig.~\ref{fig:spiderAswitch}. Therefore suppose $M$ has no cross edges, but has a pair of \emph{parallel} edges $\{u_{1j}u_{1,[j]+1},\,u_{2j}u_{2,[j]+1}\}$, for some $j\in[6]$. Switching the quadrangle $(u_{1j},u_{2j},u_{2,[j]+1},u_{1,[j]+1})$ results in a matching $M'$ with a cross edge $(u_{1j},u_{2j})$. Hence $M$ is again connected to $M_0$, via $M'$. Now any perfect matching in $A_1$ which has no cross edges or parallel edges is either $M_1$ or $M_2$, and these have no available switch.
\end{proof}
We will use this to show that $W_k$ is ergodic.
\begin{lemma}\label{lem:ergodicWk}
\G{W_k} is connected, for all $k\geq 1$.
\end{lemma}
\begin{proof}
  We use induction on $k$. As basis, $W_1$ is ergodic: \G{W_1} is shown in Fig.~\ref{fig:spiderW1}.  For $k>1$, let $X$, $Y$ be any two perfect matchings in $W_k$. From Lemma~\ref{lem:ergodicA1}, we can exchange $X\cap A_{k-1}$ and $Y\cap A_{k-1}$ to give matchings $X_1,\,Y_1$ so that $X_1\cap A_{k-1},\,Y_1\cap A_{k-1}$ have no cross edges. (See Figs.~\ref{fig:spiderNon-ergodic} and~\ref{fig:spiderAswitch}). By induction,  we can exchange $X_1\cap W_{k-1}$ to $Y_1\cap W_{k-1}$ to give matchings $X_2$, $Y_2$ so that every edge of $X_2\cap C_{k-1}$ is parallel to an edge of $X_2\cap C_k$, and every edge of $Y_2\cap C_{k-1}$ is parallel to an edge of $Y_2\cap C_k$. Using Lemma~\ref{lem:ergodicA1} again, $X_2$, $Y_2$ can be transformed to $X_3$, $Y_3$,  so that $X_3\cap A_{k-1}=Y_3\cap A_{k-1}$. Finally, by induction, $X_3$, $Y_3$ can be transformed to $X_4$, $Y_4$,  so that $X_4\cap W_{k-1}=Y_4\cap W_{k-1}$ and $X_4\cap C_k=Y_4\cap C_k$, and we are done.
\end{proof}

  \begin{figure}[htbp]
\begin{center}
  \begin{tikzpicture}[scale=0.45, line width=1pt,gray]
    \foreach \angle in {0,60,120,180,240,300}
          \node[vertex] at (\angle:2) (2\angle) {};
    \foreach \a/\b in {0/60, 120/180, 240/300} \draw (2\a)--(2\b);
    \foreach \a/\b in {60/120, 180/240, 300/0} \draw[black,line width=1.5pt] (2\a)--(2\b);
    \draw (20)--(2180);
  \end{tikzpicture} \raisebox{7.5mm}{\Large\ $\longleftrightarrow$\ }
  \begin{tikzpicture}[scale=0.45, line width=1pt,gray]
    \foreach \angle in {0,60,120,180,240,300}
          \node[vertex] at (\angle:2) (2\angle) {};
    \foreach \a/\b in {0/60, 120/180, 180/240, 300/0} \draw (2\a)--(2\b);
    \foreach \a/\b in {60/120, 240/300} \draw[black,line width=1.5pt] (2\a)--(2\b);
    \draw[black,line width=1.5pt] (20)--(2180);
  \end{tikzpicture}\raisebox{7.5mm}{\large\ $\longleftrightarrow$\ }
  \begin{tikzpicture}[scale=0.45, line width=1pt,gray]
    \foreach \angle in {0,60,120,180,240,300}
          \node[vertex] at (\angle:2) (2\angle) {};
    \foreach \a/\b in {0/60, 120/180, 240/300} \draw[black,line width=1.5pt] (2\a)--(2\b);
    \foreach \a/\b in {60/120, 180/240, 300/0} \draw (2\a)--(2\b);
    \draw (20)--(2180);
  \end{tikzpicture}
\end{center}\vspace{-0.5\baselineskip}
\caption{Switching $W_1$}\label{fig:spiderW1}
\end{figure}
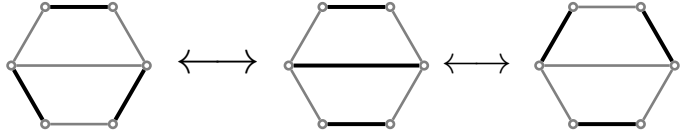

Let $W'_k=W_k\thsp\setminus\thsp u_{11}u_{14}$ be the graph $W_k$ after deleting the edge $u_{11}u_{14}$. Then the two perfect matchings $M_1$, $M_2$ in $W'_k$ given by \[M_p=\{w_{i,j}w_{i,[j]+1}:i+j=p-1 \bmod 2,\ i\in[k],\,j\in[6]\} \quad (p=1,2),\]
and shown in Fig.~\ref{fig:spiderXY}, have no available switch, so are isolated vertices in \G{W'_k}. Thus $W'_k$ is not ergodic.
Given this, we might suppose that $M_1$ and $M_2$ are far apart in \G{W_k}, and that is the case.
\begin{figure}[htbp]
\begin{center}
  \begin{tikzpicture}[scale=0.4, line width=1pt,gray]
    \foreach \angle in {0,60,120,180,240,300}
       \foreach \radius in {1,2,3,4,5}
          \node[vertex] at (\angle:\radius) (z\radius\angle) {};
    \foreach \angle in {0,60,120,180,240,300}
       \foreach[count=\ind from 1] \radius in {2,3,4,5}
          \draw (z\ind\angle)--(z\radius\angle);
    \foreach \a/\b in {0/60, 120/180, 240/300}
       \foreach \radius in {2,4}
          \draw (z\radius\a)--(z\radius\b);
    \foreach \a/\b in {0/60, 120/180, 240/300}
       \foreach \radius in {1,3,5}
          \draw[black,line width=1.5pt] (z\radius\a)--(z\radius\b);
    \foreach \a/\b in {60/120, 180/240, 300/0}
       \foreach \radius in {2,4}
          \draw[black,line width=1.5pt] (z\radius\a)--(z\radius\b);
    \foreach \a/\b in {60/120, 180/240, 300/0}
       \foreach \radius in {1,3,5}
          \draw (z\radius\a)--(z\radius\b);
          \draw (z10)--(z1180);
    \node[black] at (0,-5.5) {$M_1$};
  \end{tikzpicture} \hspace*{1cm}
\begin{tikzpicture}[scale=0.4, line width=1pt,gray]
       \foreach \angle in {0,60,120,180,240,300}
       \foreach \radius in {1,2,3,4,5}
          \node[vertex] at (\angle:\radius) (z\radius\angle) {};
    \foreach \angle in {0,60,120,180,240,300}
       \foreach[count=\ind from 1] \radius in {2,3,4,5}
          \draw (z\ind\angle)--(z\radius\angle);
    \foreach \a/\b in {0/60, 120/180, 240/300}
       \foreach \radius in {2,4}
          \draw[black,line width=1.5pt] (z\radius\a)--(z\radius\b);
    \foreach \a/\b in {0/60, 120/180, 240/300}
       \foreach \radius in {1,3,5}
          \draw (z\radius\a)--(z\radius\b);
    \foreach \a/\b in {60/120, 180/240, 300/0}
       \foreach \radius in {2,4}
          \draw (z\radius\a)--(z\radius\b);
    \foreach \a/\b in {60/120, 180/240, 300/0}
       \foreach \radius in {1,3,5}
          \draw[black,line width=1.5pt] (z\radius\a)--(z\radius\b);
          \draw (z10)--(z1180);
          \node[black] at (0,-5.5) {$M_2$};
  \end{tikzpicture}
\end{center}\vspace{-\baselineskip}
\caption{Matchings $M_1$, $M_2$ in $W_5$}\label{fig:spiderXY}
\end{figure}
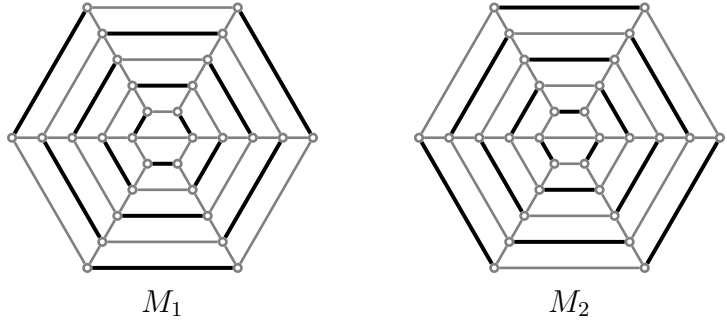

\begin{lemma}\label{lem:distanceM1M2}
The distance between $M_1$, $M_2$ in \G{W_k} is $k(3k-1)$.
\end{lemma}
\begin{proof}
Let $X_t$ be the perfect matching at step $t$ on a path $P$ from $M_1$ to $M_2$ in \G{W_k}, so $X_0=M_1$ and $X_\ell=M_2$, where $\ell$ is the path length. For any hexagon $C_i$ $(i\in[k])$, we will write $C_i(t)\supset M_j$ to mean $C_i\cap X_t=C_i\cap M_j$ $(j=1,2)$. Initially $C_i(0)\supset M_1$, for all $i\in[k]$. At step $t$, we will say $C_i$ has been \emph{exchanged} if $C_i(t)\supset M_2$. Let $s(t)=|\{i\in[k]: C_i(t)\supset M_2\}|$ denote the number of exchanged hexagons, so $s(0)=0$ and $s(\ell)=k$.

Let $t_i$ be the first step on $P$  at which $C_i$ has been exchanged, and let $t'_i<t_i$ be the last step  before any edge of $C_i$ has been  switched. Initially, the only switch that can be performed is in $W_1$, using $u_{11}u_{14}$. After two switches, $C_1$ can be exchanged (see Fig.~\ref{fig:spiderW1}). Since at least two quadrangles must be switched to change the state of six edges, this is clearly the minimum number of switches needed to exchange $C_1$. Thus $t'_1=0$, $t_1=2$ and $s(2)=1$.

For $i>1$, since $M_1$, $M_2$ are edge-disjoint, we can exchange $C_i$ only by switching all six edges.
Therefore, since no two edges of $C_i$ share a quadrangle, at least six switches are needed to exchange $C_i$.
Now, an edge of $C_i$ can be switched only if there is a parallel edge in $C_{i-1}$ or $C_{i+1}$. Since, by assumption $C_i(t'_i), C_{i+1}(t'_i)\supset M_1$, we must have $C_{i-1}(t'_i)\supset M_2$.

Then we have the situation shown in Fig~\ref{fig:spiderAswitch}, and we can perform exactly six switches in $A_{i-1}$ so that $C_i(t_i)\supset M_2$, where $t_i=t'_i+6$. However, we now have $C_{i-1}(t_i)\supset M_1$, so $s(t_i)=s(t'_i)$. Thus $s(t)$ changes only when $C_1$ is exchanged. So $C_1$ must be exchanged at least $k$ times to switch the whole of $W_k$.

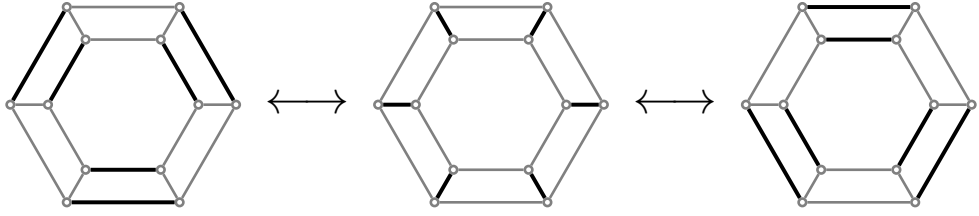
\begin{figure}[htbp]
\begin{center}
  \begin{tikzpicture}[scale=0.5, line width=1pt,gray]
    \foreach \angle in {0,60,120,180,240,300}
       \foreach \radius in {2,3}
          \node[vertex] at (\angle:\radius) (z\radius\angle) {};
    \foreach \angle in {0,60,120,180,240,300}
       \foreach[count=\ind from 2] \radius in {3}
          \draw (z\ind\angle)--(z\radius\angle);
    \foreach \a/\b in {60/120, 180/240, 300/0}
       \foreach \radius in {2}
          \draw (z\radius\a)--(z\radius\b);
    \foreach \a/\b in {0/60, 120/180, 240/300}
       \foreach \radius in {3}
          \draw[black,line width=1.5pt] (z\radius\a)--(z\radius\b);
    \foreach \a/\b in {0/60, 120/180, 240/300}
       \foreach \radius in {2}
          \draw[black,line width=1.5pt] (z\radius\a)--(z\radius\b);
    \foreach \a/\b in {60/120, 180/240, 300/0}
       \foreach \radius in {3}
          \draw (z\radius\a)--(z\radius\b);
  \end{tikzpicture} \raisebox{12.5mm}{\Large\ $\longleftrightarrow$\ }
\begin{tikzpicture}[scale=0.5, line width=1pt,gray]
    \foreach \angle in {0,60,120,180,240,300}
       \foreach \radius in {2,3}
          \node[vertex] at (\angle:\radius) (z\radius\angle) {};
    \foreach \angle in {0,60,120,180,240,300}
       \foreach[count=\ind from 2] \radius in {3}
          \draw[black,line width=1.5pt] (z\ind\angle)--(z\radius\angle);
    \foreach \a/\b in {0/60, 120/180, 240/300}
       \foreach \radius in {2}
          \draw (z\radius\a)--(z\radius\b);
    \foreach \a/\b in {0/60, 120/180, 240/300}
       \foreach \radius in {3}
          \draw (z\radius\a)--(z\radius\b);
    \foreach \a/\b in {60/120, 180/240, 300/0}
       \foreach \radius in {2}
          \draw (z\radius\a)--(z\radius\b);
    \foreach \a/\b in {60/120, 180/240, 300/0}
       \foreach \radius in {3}
          \draw (z\radius\a)--(z\radius\b);
  \end{tikzpicture}
 \raisebox{12.5mm}{\Large\ $\longleftrightarrow$\ }
\begin{tikzpicture}[scale=0.5, line width=1pt,gray]
    \foreach \angle in {0,60,120,180,240,300}
       \foreach \radius in {2,3}
          \node[vertex] at (\angle:\radius) (z\radius\angle) {};
    \foreach \angle in {0,60,120,180,240,300}
       \foreach[count=\ind from 2] \radius in {3}
          \draw (z\ind\angle)--(z\radius\angle);
    \foreach \a/\b in {0/60, 120/180, 240/300}
       \foreach \radius in {2}
          \draw (z\radius\a)--(z\radius\b);
    \foreach \a/\b in {0/60, 120/180, 240/300}
       \foreach \radius in {3}
          \draw (z\radius\a)--(z\radius\b);
    \foreach \a/\b in {60/120, 180/240, 300/0}
       \foreach \radius in {2}
          \draw[black,line width=1.5pt] (z\radius\a)--(z\radius\b);
    \foreach \a/\b in {60/120, 180/240, 300/0}
       \foreach \radius in {3}
          \draw[black,line width=1.5pt] (z\radius\a)--(z\radius\b);
  \end{tikzpicture}
\end{center}\vspace{-0.5\baselineskip}
\caption{Switching $A_i$}\label{fig:spiderAswitch}
\end{figure}

After switching $C_i$ ($i\in[k-1]$) for the first time, we have $C_i(t_i)\supset M_2$, $C_{i+1}(t_i) \supset M_1$. So we can exchange $C_{i+1}$, using six switches in $A_i$. Thus we can propagate the exchanged cycle $C_j(t)\supset M_2$ outwards, starting with $j=1$, and until $j=i$, leaving $C_j(t_i)\supset M_1$ ($j\in[i-1]$). See Fig.~\ref{fig:spiderC5}.

Since $C_k$ must be switched, we continue this outward propagation until $i=k$. Then we have $C_k(t_k)\supset M_2$ and $C_k(t_k)\supset M_1$ ($i\in[k-1]$), after $t_k=6(k-1)+2=6k-4$ switches. This is clearly the minimum number of switches needed to exchange $C_k$, starting from $X_0=M_1$.

\begin{figure}[htbp]
\begin{center}
\hspace*{15mm}\begin{tikzpicture}[scale=0.35, line width=1pt,gray]
    \foreach \angle in {0,60,120,180,240,300}
       \foreach \radius in {1,2,3,4,5}
          \node[vertex] at (\angle:\radius) (z\radius\angle) {};
    \foreach \angle in {0,60,120,180,240,300}
       \foreach[count=\ind from 1] \radius in {2,3,4,5}
          \draw (z\ind\angle)--(z\radius\angle);
    \foreach \a/\b in {0/60, 120/180, 240/300}
       \foreach \radius in {2,4,5}
          \draw (z\radius\a)--(z\radius\b);
    \foreach \a/\b in {0/60, 120/180, 240/300}
       \foreach \radius in {3,5}
          \draw (z\radius\a)--(z\radius\b);
    \foreach \a/\b in {60/120, 180/240, 300/0}
       \foreach \radius in {1,3,5}
          \draw[black,line width=1.5pt] (z\radius\a)--(z\radius\b);
    \foreach \a/\b in {0/60, 60/120, 120/180, 180/240, 240/300, 300/0}
       \foreach \radius in {1,2,3,4}
          \draw (z\radius\a)--(z\radius\b);
          \foreach \a/\b in {0/60, 120/180, 240/300}
       \foreach \radius in {2,4}
          \draw[black,line width=1.5pt] (z\radius\a)--(z\radius\b);
    \draw (z10)--(z1180);
   \end{tikzpicture}\raisebox{15mm}{{\Large\ $\stackrel{W_1}{\longleftrightarrow}$\ }}
\begin{tikzpicture}[scale=0.35, line width=1pt,gray]
    \foreach \angle in {0,60,120,180,240,300}
       \foreach \radius in {1,2,3,4,5}
          \node[vertex] at (\angle:\radius) (z\radius\angle) {};
    \foreach \angle in {0,60,120,180,240,300}
       \foreach[count=\ind from 1] \radius in {2,3,4,5}
          \draw (z\ind\angle)--(z\radius\angle);
    \foreach \a/\b in {0/60, 120/180, 240/300}
       \foreach \radius in {2,4,5}
          \draw (z\radius\a)--(z\radius\b);
    \foreach \a/\b in {0/60, 120/180, 240/300}
       \foreach \radius in {3,5}
          \draw (z\radius\a)--(z\radius\b);
    \foreach \a/\b in {60/120, 180/240, 300/0}
       \foreach \radius in {3,5}
          \draw[black,line width=1.5pt] (z\radius\a)--(z\radius\b);
    \foreach \a/\b in {0/60, 60/120, 120/180, 180/240, 240/300, 300/0}
       \foreach \radius in {1,2,3,4}
          \draw (z\radius\a)--(z\radius\b);
          \foreach \a/\b in {0/60, 120/180, 240/300}
       \foreach \radius in {1,2,4}
          \draw[black,line width=1.5pt] (z\radius\a)--(z\radius\b);
    \draw (z10)--(z1180);
   \end{tikzpicture}
   \raisebox{15mm}{\Large\ $\stackrel{A_1}{\longleftrightarrow}$\ }
  \begin{tikzpicture}[scale=0.35, line width=1pt,gray]
    \foreach \angle in {0,60,120,180,240,300}
       \foreach \radius in {1,2,3,4,5}
          \node[vertex] at (\angle:\radius) (z\radius\angle) {};
    \foreach \angle in {0,60,120,180,240,300}
       \foreach[count=\ind from 1] \radius in {2,3,4,5}
          \draw (z\ind\angle)--(z\radius\angle);
    \foreach \a/\b in {0/60, 120/180, 240/300}
       \foreach \radius in {1,2,4,5}
          \draw (z\radius\a)--(z\radius\b);
    \foreach \a/\b in {0/60, 120/180, 240/300}
       \foreach \radius in {3,5}
          \draw (z\radius\a)--(z\radius\b);
    \foreach \a/\b in {60/120, 180/240, 300/0}
       \foreach \radius in {1,2,3,5}
          \draw[black,line width=1.5pt] (z\radius\a)--(z\radius\b);
    \foreach \a/\b in {0/60, 60/120, 120/180, 180/240, 240/300, 300/0}
       \foreach \radius in {2,3,4}
          \draw (z\radius\a)--(z\radius\b);
          \foreach \a/\b in {0/60, 120/180, 240/300}
       \foreach \radius in {4}
          \draw[black,line width=1.5pt] (z\radius\a)--(z\radius\b);
       \draw (z10)--(z1180);
  \end{tikzpicture}\\[5mm]
   \raisebox{15mm}{\Large\ $\stackrel{A_2}{\longleftrightarrow}$\ }
\begin{tikzpicture}[scale=0.35, line width=1pt,gray]
    \foreach \angle in {0,60,120,180,240,300}
       \foreach \radius in {1,2,3,4,5}
          \node[vertex] at (\angle:\radius) (z\radius\angle) {};
    \foreach \angle in {0,60,120,180,240,300}
       \foreach[count=\ind from 1] \radius in {2,3,4,5}
          \draw (z\ind\angle)--(z\radius\angle);
    \foreach \a/\b in {0/60, 120/180, 240/300}
       \foreach \radius in {2,4,5}
          \draw (z\radius\a)--(z\radius\b);
    \foreach \a/\b in {0/60, 120/180, 240/300}
       \foreach \radius in {1,3,5}
          \draw (z\radius\a)--(z\radius\b);
    \foreach \a/\b in {60/120, 180/240, 300/0}
       \foreach \radius in {1,5}
          \draw[black,line width=1.5pt] (z\radius\a)--(z\radius\b);
    \foreach \a/\b in {0/60, 60/120, 120/180, 180/240, 240/300, 300/0}
       \foreach \radius in {2,3,4}
          \draw (z\radius\a)--(z\radius\b);
          \foreach \a/\b in {0/60, 120/180, 240/300}
       \foreach \radius in {2,3,4}
          \draw[black,line width=1.5pt] (z\radius\a)--(z\radius\b);
       \draw (z10)--(z1180);
  \end{tikzpicture}
  \raisebox{15mm}{\Large\ $\stackrel{A_3}{\longleftrightarrow}$\ }
  \begin{tikzpicture}[scale=0.35, line width=1pt,gray]
    \foreach \angle in {0,60,120,180,240,300}
       \foreach \radius in {1,2,3,4,5}
          \node[vertex] at (\angle:\radius) (z\radius\angle) {};
    \foreach \angle in {0,60,120,180,240,300}
       \foreach[count=\ind from 1] \radius in {2,3,4,5}
          \draw (z\ind\angle)--(z\radius\angle);
    \foreach \a/\b in {0/60, 120/180, 240/300}
       \foreach \radius in {2,4,5}
          \draw (z\radius\a)--(z\radius\b);
    \foreach \a/\b in {0/60, 120/180, 240/300}
       \foreach \radius in {1,3,5}
          \draw (z\radius\a)--(z\radius\b);
    \foreach \a/\b in {60/120, 180/240, 300/0}
       \foreach \radius in {1,3,4,5}
          \draw[black,line width=1.5pt] (z\radius\a)--(z\radius\b);
    \foreach \a/\b in {0/60, 60/120, 120/180, 180/240, 240/300, 300/0}
       \foreach \radius in {2}
          \draw (z\radius\a)--(z\radius\b);
          \foreach \a/\b in {0/60, 120/180, 240/300}
       \foreach \radius in {2}
          \draw[black,line width=1.5pt] (z\radius\a)--(z\radius\b);
       \draw (z10)--(z1180);
  \end{tikzpicture}\raisebox{15mm}{\Large\ \,$\stackrel{A_4}{\longleftrightarrow}$\ }
\begin{tikzpicture}[scale=0.35, line width=1pt,gray]
    \foreach \angle in {0,60,120,180,240,300}
       \foreach \radius in {1,2,3,4,5}
          \node[vertex] at (\angle:\radius) (z\radius\angle) {};
    \foreach \angle in {0,60,120,180,240,300}
       \foreach[count=\ind from 1] \radius in {2,3,4,5}
          \draw (z\ind\angle)--(z\radius\angle);
    \foreach \a/\b in {0/60, 120/180, 240/300}
       \foreach \radius in {2,4,5}
          \draw[black,line width=1.5pt] (z\radius\a)--(z\radius\b);
    \foreach \a/\b in {0/60, 120/180, 240/300}
       \foreach \radius in {1,3,5}
          \draw (z\radius\a)--(z\radius\b);
    \foreach \a/\b in {60/120, 180/240, 300/0}
       \foreach \radius in {2,4,5}
          \draw (z\radius\a)--(z\radius\b);
    \foreach \a/\b in {60/120, 180/240, 300/0}
       \foreach \radius in {1,3}
          \draw[black,line width=1.5pt] (z\radius\a)--(z\radius\b);
    \draw (z10)--(z1180);
  \end{tikzpicture}
\end{center}
\caption{Exchanging $C_5$}\label{fig:spiderC5}
\end{figure}
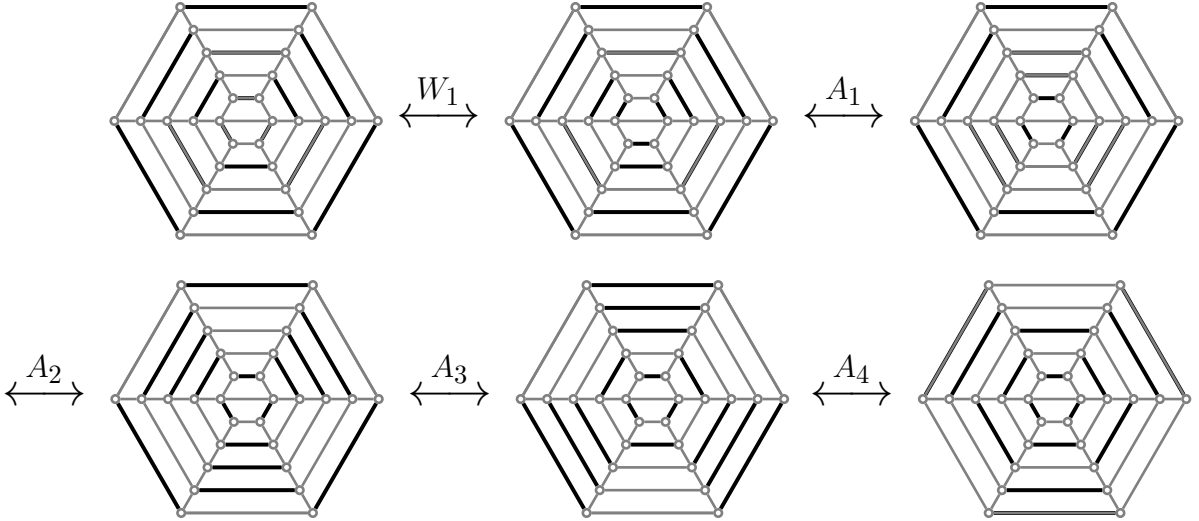

Now, for each $i=k,k-1,\ldots,2,1$, suppose we have $M_1$ in $W_i$ and $M_2$ in $W_k\setminus W_i$.  Then we can exchange $C_{i}$ in $W_i$ as above, and hence $P$ will terminate  with $X_\ell=M_2$. The minimum number of switches needed to exchange $C_i$ in $W_i$ is $t_i=(6i-4)$. Let $\ell_{\min}$ be the minimum total number of switches required  to exchange all $C_i$ ($i\in[k]$). It follows that $\ell_{\min}=\sum_{i=1}^k t_i =\sum_{i=1}^k (6i-4)=k(3k-1)$. Thus $k(3k-1)$ is the minimum length of any path $P$ in \G{G} from $M_1$ to $M_2$.
\end{proof}

\begin{theorem}\label{thm:diametern2}
There exists a sequence of graphs $G_n$, on $n$ vertices, such that the transition graph \G{G_n} is connected, but has diameter $\Omega(n^2)$.
\end{theorem}
\begin{proof}
For the sequence of graphs $W_k$ ($k=1,2,\ldots$), we have $n=6k$ and, from Lemma~\ref{lem:distanceM1M2}, \G{W_k} has diameter at least $k(3k-1)= n(n-2)/12$.
\end{proof}

However, in spite of having no sub-exponential bound on the diameter, it not clear that we can even construct graphs for which the diameter is $\Omega(n^3)$.

\section{Exact counting}\label{sec:exact}
We conclude this paper by considering the problem of exactly counting perfect matchings in some ``small'' classes of graphs.
For graphs in such classes, there is often an ordering which permits a dynamic programming type of algorithm to be employed. Such an algorithm was given in~\cite{DyJeMu17}, for example,  for monotone graphs of small width. Here we give algorithms for three graph classes defined in the Appendix.

\subsection{Cographs}

For two graphs $G=(V,E)$ and $H=(W,F)$ with $V \cap W = \emptyset$ we define
their \emph{disjoint union} $G \uplus  H = (V \cup W, E \cup F)$,
and \emph{complete join}
$G \Join H = (V \cup W, E \cup F \cup \{ vw \mid v \in V, w \in W\})$. These
two operations are complementary:
$\ol{G \Join H} = \ol{G} \uplus   \ol{H}$ and $\ol{G \uplus   H} = \ol{G} \Join \ol{H}$.

A graph $G$ is a \emph{cograph} (or \emph{complement-reducible}) if
\begin{enumerate}[itemsep=0pt,label=(\alph*)]
\item $G \simeq K_1$, that is, $G$ has one vertex and no edges, or
\item $G=G_1\uplus G_2$, where $G_1,G_2$ are cographs, or\label{cplus}
\item $G=G_1\Join G_2$, where $G_1,G_2$ are cographs.\label{cjoin}
\end{enumerate}

The class of cographs was introduced in \cite{CoLeBu81}. In particular,
it was shown in~\cite{CoLeBu81} that $G$ is a cograph if and only it is $P_4$-free,
where $P_4$ is the path with four vertices and three edges.
Since $P_4=\ol{P_4}$, this implies that $G$ is a cograph if and only if $\ol{G}$ is a cograph.

The decomposition of a cograph can be represented by a rooted binary tree $T$,
called a \emph{cotree}. The leaves of $T$ are the vertices of $G$,
and its internal nodes are marked $\uplus $ and $\Join$, corresponding to
constructions \ref{cplus} and  \ref{cjoin} above. Two vertices are
adjacent in $G$ if and only if their lowest common ancestor in $T$ is
marked $\Join$.
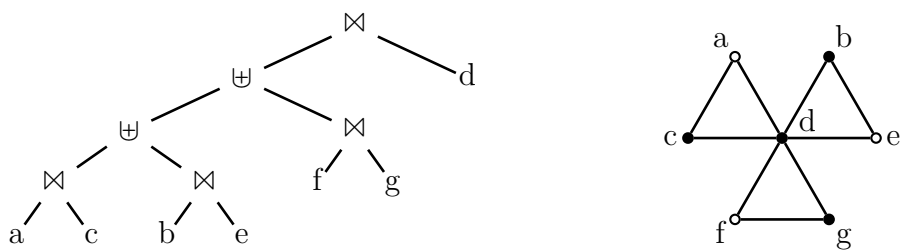
\begin{figure}[htbp]
  \newcommand{\st}[1]{\makebox(0,0){\rule[-4pt]{0pt}{13pt}#1}}
  \tikzset{b/.style = {circle,draw,inner sep=0pt,fill,  minimum size=1.2mm},
           w/.style = {circle,draw,inner sep=0pt,thick, minimum size=1.25mm}
     }
  \hspace*{\fill}
  \begin{tikzpicture}[xscale=0.5, yscale=0.7]
    \node (a) at ( 1,1) {\st a}; \node (c) at ( 3,1) {\st c};
    \node (b) at ( 5,1) {\st b}; \node (e) at ( 7,1) {\st e};
    \node (f) at ( 9,2) {\st f}; \node (g) at (11,2) {\st g};
    \node (d) at (13,4) {\st d};
    \node (ac) at ( 2,2) {$\Join$};  \draw  (a)--(ac)--(c);
    \node (be) at ( 6,2) {$\Join$};  \draw  (b)--(be)--(e);
    \node (fg) at (10,3) {$\Join$};  \draw  (f)--(fg)--(g);
    \node (ab) at ( 4,3) {$\uplus$}; \draw (ac)--(ab)--(be);
    \node (af) at ( 7,4) {$\uplus$}; \draw (ab)--(af)--(fg);
    \node (ad) at (10,5) {$\Join$};  \draw (af)--(ad)--(d);
  \end{tikzpicture}
  \hspace*{\fill}
  \begin{tikzpicture}[scale=1.25]
    \node[w, label=120:\st a] (a) at (120:1) {};
    \node[b, label= 60:\st b] (b) at ( 60:1) {};
    \node[b, label=180:\st c] (c) at (180:1) {};
    \node[w, label=  0:\st e] (e) at (  0:1) {};
    \node[w, label=240:\st f] (f) at (240:1) {};
    \node[b, label=300:\st g] (g) at (300:1) {};
    \tikzset{label distance=3pt}
    \node[b, label= 10:\st d] (d) at (  0,0) {};
    \draw (d)--(a)--(c)--(d)--(b)--(e)--(d)--(f)--(g)--(d);
  \end{tikzpicture}
  \hspace*{\fill}
  \caption{A cotree and the cograph represented with a bipartition into a tripod.}
  \label{fig:cograph}
\end{figure}

\subsubsection*{Recurrence equation}

For a graph $G$ and an integer $s$ let $m(G,s)$ denote the number of
matchings of $G$ that have size exactly $s$. If $G$ has $n$ vertices then
$m(G,s)=0$ holds for $s<0$ or $2s>n$. For cographs the values of $m$
can be computed recursively as follows:

\begin{description}
\item[leaf] For a cograph with one vertex we have $m(K_1,s) = 1$ if $s=0$
  and $m(K_1,s) = 0$ otherwise, because $K_1$ has only one matching, the
  empty set.
\item[union] For the disjoint union of two graphs $G$ and $H$ we have
  \[ m(G \uplus  H, s) = \sum_{i=0}^s m(G,i) \cdot m(H, s-i) \]
  because there are no edges between the vertices of $G$ and
  the vertices of $H$.
\item[join]
% The complete join $G \Join H$ of two graphs $G$ and $H$ is
% $\ol{\ol{G} \uplus  \ol{H}}$.
  For $g = |V(G)|$ and $h = |V(H)|$ we have
  \[  m(G \Join H, s) = \sum_{i=0}^{\min(s,  \lfloor g/2\rfloor)} \hq
                          \sum_{j=0}^{\min(s-i,\lfloor h/2\rfloor)} \hm
      m(G,i) \cdot m(H,j) \cdot\binom{g-2i}{k}\cdot\binom{h-2j}{k}\cdot k! \]
  where $k=s-i-j$. A matching of size $s$ in $G \Join H$ partitions into
  a matching of size $i$ in $G$, a matching of size $j$ in $H$ and a matching
  of size $k$ between the vertices of $G$ and $H$. The subgraph $G$ has
  $g-2i$ vertices that are not matched internally, that is, are unsaturated
  or matched to vertices of $H$. Similarly, the subgraph $H$ has
  $h-2j$ vertices that are not matched internally. From both sets we
  can choose exactly $k$ vertices to be matched across the join.
  The partial subgraph isomorphic to $K_{k,k}$ has $k!$ perfect matchings.
\end{description}

\subsubsection*{Algorithm}

Let $G$ be a cograph with $n$ vertices. A cotree $T$ of $G$ can be computed
in linear time \cite{HabPau45} and has $n-1$ internal nodes.  For every
cograph $H$ represented by a rooted subtree of $T$ we can compute $m(H,s)$
for all values of $s$ with $0 \le s \le n/2$. This takes time $O(1)$ for leaves,
$O(n)$ for union nodes and $O(n^2)$ for join nodes. Hence $m(G,n/2)$,
the number of perfect matchings of $G$, can be computed in time $O(n^4)$.

\subsection{Graphs with bounded treewidth}

% \begin{definition} \label{def:tw}
  A pair $(T,X)$ is a \emph{tree decomposition} of a graph $G=(V,E)$
  if $T$ is a tree with node set $I$ and $X$ maps nodes of $T$ to
  subsets of $V$ (called \emph{bags}) such that
  \begin{enumerate}
  \item \label{tw:1} % $\bigcup_{i \in I} X(i) = V$
    $\forall v \in V ,\hq \exists i \in I, \hq v \in X(i)$;
  \item \label{tw:2}
    $\forall uv \in E, \hq \exists i \in I, \hq \{u,v\} \subseteq X(i)$;
  \item \label{tw:3}
    $\forall v \in V, \hq  T[\{i \mid v \in X(i)\}]$ is connected.
  \end{enumerate}
  The \emph{width} of $(T,X)$ is $\max_{i \in I} |X(i)|-1$ and the
  \emph{treewidth} of $G$ is the minimum width of a tree decomposition of $G$.
  It is denoted as $\tw(G)$.
% \end{definition}
The class of graphs with $\tw(G)\leq w$, for some constant $w$, is clearly hereditary.

%% A tree decomposition is \emph{small} if $X(i) \sm X(j) \ne \emptyset$ and $X(j)
%% \sm X(i) \ne \emptyset$ holds for every pair of different nodes $i$ and $j$.
%% Every graph $G$ has a small tree decomposition of width $\tw(G)$.
%% For every small tree decomposition the number of nodes of the tree
%% does not exceed the number of vertices of the graph.\footnote{Is this definition used?}

In a \emph{rooted} tree decomposition we choose one node $r$ to become
the root of the tree. For all other nodes $i$, the neighbour of $i$ on
the path to $r$ is the \emph{parent} of $i$, all other neighbours are
\emph{children} of $i$. All neighbours of $r$ are children of the root.
For a rooted tree decomposition $(T,X)$ and every $i \in I$ let
$Y(i) = X(i) \cup \bigcup_{j} Y(j)$ where the union is taken over all
children of $i$. Especially we have $Y(i)=X(i)$ for all leaves $i$ of $T$,
and $Y(r)=V$.

A \emph{nice} tree decomposition of $G=(V,E)$ is a rooted tree
decomposition $(T,X)$ of $G$ where each node has at most two children,
which recursively uses the operations:
\begin{description}
\item[start] If $i$ is a leaf of $T$ then $X(i)=\emptyset$.
\item[introduce/forget] If $i$ has exactly one child $j$ then
% $|(X(i) \sm X(j))\cup(X(j) \sm X(i))|=1$.
  $X(i)$ and $X(j)$ differ by one vertex. More precisely, $i$ is an
  \emph{introduce} node if $X(i) \supset X(j)$ and $i$ is a \emph{forget} node
  if $X(i) \subset X(j)$.
\item[join] If $i$ has two children $j$ and $k$ then $X(i)=X(j)$ and
  $X(i)=X(k)$.
\item[root] The root $r$ is a node with $X(r)=\emptyset$, usually a forget
  node, but a start node if $V=\emptyset$.
\end{description}
Every graph $G=(V,E)$ has a nice tree decomposition of width $\tw(G)$
that contains $O(|V|)$ % at most $4|V|$
nodes, see Lemma 13.1.2 on page 149 of \cite{Kloks94}.

\subsubsection*{Recurrence equations} \label{ss:re}

Let $(T,X)$ be a nice tree decomposition of a graph $G=(V,E)$. For every
node $i$ of $T$ and every set $U \subseteq X(i)$ let $p(i,U)$ denote
the number of perfect matchings in the graph $G[Y(i) \sm U]$ such that
every vertex in $X(i) \sm U$ is matched to a vertex in $Y(i) \sm X(i)$.
That is, a matching containing an edge with both endpoints in $X(i)$
does not contribute to $p(i,U)$ for any $U$.
The numbers $p(i,U)$ can be computed recursively as follows:

\begin{description}
\item[start] If $i$ is a leaf of $T$ then $p(i,\emptyset)=1$ since $\emptyset$ is
  the unique perfect matching of the empty graph.
\item[introduce] If $i$ is an introduce node with child $j$ and
  $v \in X(i) \sm X(j)$ then $p(i,U) = 0$ and
  $p(i,U\cup\{v\}) = p(j,U)$ hold for all $U \subseteq X(j)$.
  By Condition \ref{tw:2} of the definition % \ref{def:tw}
  the new vertex $v$ has no neighbour vertex in $Y(i) \sm X(i)$.
\item[forget] If $i$ is a forget node with child $j$ and $v \in X(j) \sm X(i)$
  then
  \[ p(i,U) = \sum_{u \in \nbh{v} \cap X(i) \sm U} p(j,U\cup\{u,v\}) \]
  holds for all $U \subseteq X(i)$.
  The vertex $v \in X(j) \sm X(i)$ must be matched to a neighbour $u \in X(i)$.
  We add the edge $uv$ to the matchings of $G[Y(j) \sm (U\cup\{u,v\})]$.
  Note that $p(i,U)=0$ if $\nbh{v} \cap X(i) \sm U = \emptyset$.
\item[join] If $i$ is a join node with children $j$ and $k$ then
  \[ p(i,U) = \sum_{\begin{smallmatrix}
                     J \cap K = \emptyset \\ J \cup K = X(i) \sm U
                    \end{smallmatrix}} p(j,U \cup J) \cdot p(k, U \cup K) \,. \]
  Condition \ref{tw:3} of the definition % \ref{def:tw}
  implies $(Y(j) \sm X(j)) \cap (Y(k) \sm X(k)) = \emptyset$. In $G[Y(i)\sm U]$
  every matching edge with exactly one endpoint in $X(i) \sm U$ has
  its other endpoint either in $Y(j) \sm X(j)$ or in $Y(k) \sm X(k)$.
\end{description}

\subsubsection*{Algorithm}
The following generalises the algorithm given in~\cite{DyJeMu17} for bounded-degree monotone graphs.
Let $G=(V,E)$ be a graph with a nice tree decomposition
$(T,X)$ rooted at $r$. By the definition of $p(i,U)$ the graph $G$ has
$p(r,\emptyset)$ perfect matchings. This value can be computed recursively
by the recurrence equations above.  If the width of $(T,X)$ is
$w$ then such an algorithm will run in time $O(3^wn)$, where $n = |V|$,
by computing ``bottom up'' from the leaves to the root in the tree $T$.
In the case where $(T,X)$ is a path decomposition, that is, there are no
join nodes,  the algorithm takes only $O(w2^wn)$ time.

\subsection{Complements of chain graphs}

A bipartite graph $G=(V,E)$ with bipartition $(X,Y)$ is a \emph{chain graph}
if for every pair of vertices $u,v \in X$ we have $\Adj(u)\subseteq \Adj(v)$
or $\Adj(u) \supseteq \Adj(v)$. That is, the vertices in $X$ can be linearly
ordered such that $\Adj(x_1) \subseteq \Adj(x_2) \subseteq \dots
\subseteq \Adj(x_n)$. It is easy to see that this implies a linear ordering
on $Y$ as well such that $\Adj(y_1) \supseteq \Adj(y_2) \supseteq \dots
\supseteq \Adj(y_m)$.

For the sake of completeness, we re-derive a recurrence given in~\cite{DyJeMu17}
for the number of matchings in a chain graph. For positive integers $m$ and $n$
let $G=(V,E)$ be a chain graph, as defined above, with
$V = X_n \cup Y_m$ where $X_n = \{x_i \mid 1 \le i \le n\}$ and
$Y_m = \{y_j \mid 1 \le j \le m\}$. Let $M(i,s)$ be the number of matchings
of size exactly $s$ in the subgraph $G_i$ of $G$ induced by $X_i \cup Y_m$.
We have
\begin{align*}
  M(i,0) &= 1 && \text{for~} 0 \le i \le n \\
  M(i,s) &= 0 && \text{for~} 0 \le i < s \le n \\
% M(1,1) &= \deg{x_1} \\
  M(i,s) &= M(i-1,s) + (\deg{x_i}-s+1)M(i-1,s-1)
              && \text{for~} 1 \le s \le i \le n
\end{align*}
In the last equation, $M(i-1,s)$ counts matchings of size $s$ in $G_i$ with
$x_i$ unmatched. The other term counts all matchings of size $s$ in $G_i$
with $x_i$ is matched, as follows. Since $G_i$ is a chain graph, each matching
of size $(s-1)$ in $G_{i-1}$ can be extended to a matching of size $s$ in
$G_i$, with $x_i$ matched, in exactly $(\deg{x_i}-s+1)$ ways.

Next we consider complete graphs. Let $p(G)$ denote the number of perfect
matchings in $G$. Then $p(K_{2n+1}) = 0$ and $p(K_{2n}) = (2n)!!$, where
$(2n)!! = 2\cdot4\cdots(2n-2)(2n)$, which is $2^nn!$.

Finally let $G$ be the complement of a chain graph with bipartition $(X,Y)$.
Then $X$ and $Y$ are cliques of $G$, and if we remove their edges from $G$
we obtain a chain graph $G^{\mathrm{b}}$. For $|X|=n$ and $|Y|=m$ we have
\[ p(G) = \sum_{s=0}^{\min(n,m)} M(n,s) \cdot p(K_{n-s}) \cdot p(K_{m-s}) \]
where $M(n,s)$ is the number of matchings of size exactly $s$ in
$G^{\mathrm{b}}$. Since we can compute $M(i,s)$ for all values of $i$ and $s$
in $O(n^2)$ time using the recurrence above, $p(G)$ can be computed in this
time as well.

%\bibliography{monotone,martinbib,random}
%\bibliographystyle{plain}

%\newpage

\appendix
\baselineskip 12pt
\section*{Appendix: Containment of graph classes}
The hereditary classes considered above include the following. Some of these graph classes are defined (or can easily be described) by their minimal forbidden induced subgraphs. To that end,
we will define the following (non-hereditary) graph classes.
\begin{enumerate}[topsep=0pt,itemsep=-3pt,label={}]
  \item Holes:   $\IH_n  = \{C_i \mid i \in \IN, i \ge n\}$, where $C_i$ is a chordless $i$-cycle. \item Antiholes:  $\ol{\IH}_n = \{\ol{C}_i \mid i \in \IN, i \ge n\}$,  where  $\ol{C}_i$ is the complement of $C_i$.
  \item Even holes:  $\IE_n = \{C_{2i} \mid i \in \IN, 2i \ge n\}$.
  \item Odd holes:   $\IO_n = \{C_{2i+1} \mid i \in \IN, 2i+1 \ge n\}$.
  \item Suns:  $\IS_n= \{S_i \mid i \in \IN, i \ge n\}$, where $S_i$ is the $i$-sun. The graph
  $S_i$ contains a complete graph $K_i$, together a new vertex $w_e$ and edges $uw_e$, $vw_e$ for each edge $e=uv$ of a Hamilton cycle of $K_i$. The 3-sun is shown in Fig.~\ref{quasi:fig20}.
\end{enumerate}

\begin{description}[font=\sc,topsep=1pt,itemsep=0pt]
\item[\ohf] \emph{Odd hole-free graphs} are the $\IO_5$-free graphs.
\item[\ehf] \emph{Even hole-free graphs} are the $\IE_6$-free graphs. Note that
  some papers, \eg~\cite{Vuskov10}, define even hole-free graphs to
  be $\IE_4$-free.
\item[\switch] \emph{Switchable\ graphs} are defined in  Section~\ref{sec:ergodicity}.
\item[\wkch] \emph{Weakly chordal graphs}, also known as weakly
  triangulated graphs, are defined by the absence of graphs in
  $\IH_5 \cup \ol{\IH}_5$.
\item[\bip] \emph{Bipartite graphs} can be coloured by two
  colours. That is, their vertex set splits into two independent
  subsets, called colour classes or partite sets. Bipartite graphs are
  exactly the $\IO_3$-free graphs.
\item[\perf] \emph{Perfect graphs} are defined by the absence of graphs in
  $\IO_5 \cup \ol{\IO}_5$, from~\cite{ChRoST06}.
\item[\och] A graph $G$ is \emph{odd-chordal} if every even
  cycle of length at least six in $G$ has an odd chord.
\item[\chord] A graph $G$ is \emph{chordal} if every cycle of length at
  least four in $G$ has a chord. That is, chordal graphs are the
  $\IH_4$-free graphs.
\item[\cbg] A bipartite graph $G$ is \emph{chordal bipartite} if every cycle
  of length at least six in $G$ has a chord. Since every cycle in a
  bipartite graph is even, and every chord is odd, the class of
  chordal bipartite graphs is the intersection of the classes of
  odd chordal and bipartite graphs. This class is characterised by
  the forbidden set $\IO_3 \cup \IE_6$. That is, every chordless cycles in
  a chordal bipartite graph has length four.
\item[\tww] These are the classes of bounded \emph{treewidth}.
  That is, for every value of $k$ there is a class
  $\{G \mid \operatorname{tw}(G) \le k\}$. For $k=0$ this is all
  edgeless graphs, for $k=1$ all forests. For example, the permutation graph
  in Fig.~\ref{fig100} has treewidth 2.
\item[\strch] The class of \emph{strongly chordal graphs} is the
  intersection of the classes of odd chordal and chordal graphs, see
  \cite{Farber83}. This class is characterised by the forbidden set
  $\IH_4 \cup \IS_3$.
\item[\splt] The vertex set of a \emph{split graph} splits into a clique and
  an independent set. These are exactly the chordal graphs with chordal
  complement. The class is characterised by forbidden $2K_2$, $C_4$ and
  $C_5$.
\item[\conv] A bipartite graph is \emph{convex} if one of its partite sets
  can be linearly ordered such that, for each vertex in the other partite set,
  the neighbours appear consecutively.
\item[\tree] An acyclic graph is called \emph{forest}. Each connected
  component of a forest is a tree. Forests have treewidth at most one.
  Their minimal forbidden graphs are $\IH_3$.
\item[\scs] The class of \emph{strongly chordal split} graphs is the
  intersection of the classes of strongly chordal graphs and split graphs,
  characterised by the minimal forbidden subgraphs in $\IS_3 \cup \{2K_2,
  C_4, C_5\}$.
\item[\bic] A bipartite graph is \emph{biconvex} if both its partite sets
  can be linearly ordered such that all neighbourhoods appear consecutively.
\item[\perm] \emph{Permutation graphs} are the intersection
  graphs of straight line segments between two parallel lines. The
  ordering of the endpoints defines the characteristic permutation.
  The intersection model is also called matching diagram.
\item[\qmon] A graph is \emph{quasimonotone} if all its bipartitions are
  monotone.
\item[\intl] \emph{Interval graphs} are the intersection graphs of
  intervals on the real line.
\item[\mono] The class of \emph{monotone graphs} is the intersection of
  the classes of bipartite graphs and permutation graphs, see \cite{DyJeMu17}.
\item[\chp] The class of \emph{chordal permutation graphs} is the
  intersection of the classes of permutation graphs and interval graphs.
\item[\Efcb] \emph{E-free} (chordal bipartite) graphs have been
  characterised in \cite{DyeMul15}.
\item[\cogr] A graph is complement reducible, or \emph{cograph}, if it has
  at most one vertex, or is the disjoint union or the complete join of
  smaller cographs. The class is characterised by the forbidden $P_4$
  \cite{CoLeBu81}.
\item[\unii] \emph{Unit interval graphs} are the intersection graphs of
  unit-length intervals on the real line. The minimal forbidden graphs for
  this class are $\IH_4 \cup \{K_{1,3},S_3,\ol{S_3}\}$.
\item[\chain] A bipartite graph is a \emph{chain graph} if every pair of
  vertices in the same partite set has comparable neighbourhoods. This class
  is characterised by the minimal forbidden subgraphs $C_3$, $2K_2$ and
  $C_5$.
\item[\qchs] The quasi-class of disjoint unions of chain graphs.
\item[\thre] \emph{Threshold graphs}, characterised by forbidden $2K_2$,
  $C_4$ and $P_4$.
\item[Cochain] Complements of chain graphs, characterised by
  the absence of $3K_1$, $C_4$ and $C_5$.
\item[\cb] Complete bipartite graphs are characterised by the absence
  of $K_2+K_1$ (the complement of $P_3$) and $C_3$.
\item[\compl] Complete graphs are the $2K_1$-free graphs.
\item[\qcbs] Quasi-graphs of disjoint unions of complete bipartite graphs.
  The class is characterised by the absence of $P_4$, \emph{paw} (a triangle with pendant edge) and \emph{diamond} (two triangles sharing one edge).
  Every component of a graph in \qcbs\ is complete or complete bipartite~\cite{DyeMue18}.
\end{description}

The graph classes are partially ordered by inclusion.
Fig. \ref{fig:Hasse} shows a Hasse-diagram of this partial order,
restricted to the classes we consider in this paper and some others.

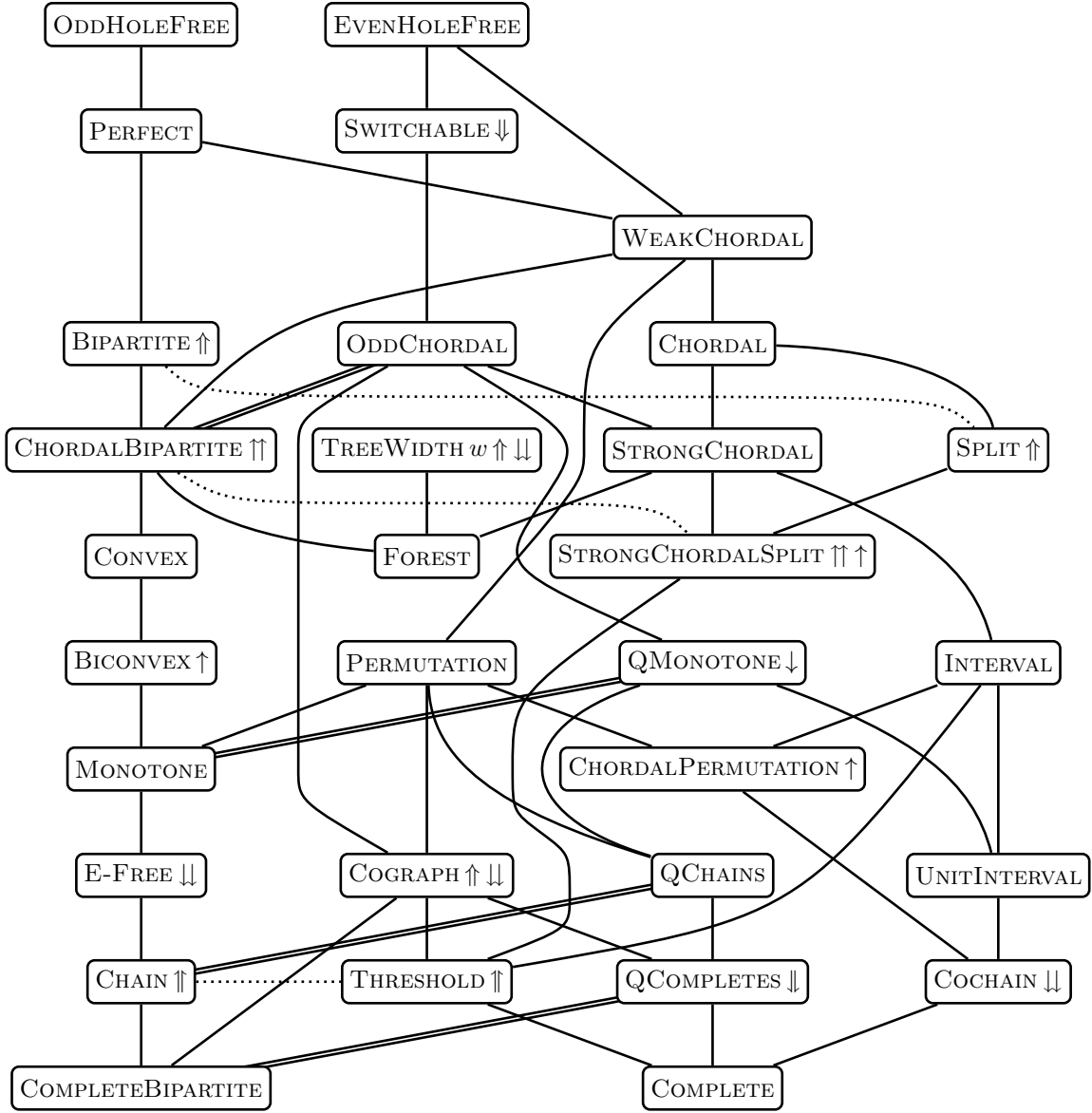
\begin{figure}[H]
  \centering
  \begin{tikzpicture}[xscale=4.0,yscale=1.49,font=\footnotesize\sc]
    \draw[double] (1,6)--(2,7)  (1,3)--(3,4)  (1,1)--(3,2)  (1,0)--(3,1);
    \node[class] at (1,10) (oh-f)  {\ohf};
    \node[class] at (2,10) (eh-f)  {\ehf};
    \node[class] at (1,9) (perf)  {\perf};
    \node[class] at (2,9) (noch)  {\switch\,\ergo}; \draw (eh-f)--(noch);
    \node[class] at (3,8) (wkch)  {\wkch};        \draw (eh-f)--(wkch)--(perf);
    \node[class] at (1,7) (bip)   {\bip\,\nerg};  \draw (oh-f)--(perf)--(bip);
    \node[class] at (2,7) (oddch) {\och};         \draw (noch)--(oddch);
    \node[class] at (3,7) (chord) {\chord};       \draw (wkch)--(chord);
    \node[class] at (1,6) (chbip) {\cbg\,\hard};  \draw (bip)--(chbip);
    \draw (wkch)  .. controls (1.5,7.3) .. (chbip);
  % \draw[double] (oddch)--(chbip);
    \node[class] at (2,6) (tw<w)  {\tww\,\nerg\,\easy};
    \node[class] at (3,6) (strch) {\strch};       \draw (oddch)--(strch)--(chord);
    \node[class] at (4,6) (split) {\splt\,\nerg};
    \draw (chord) to[bend left=40] (split);
    \node[class] at (1,5) (conv)  {\conv};        \draw (chbip)--(conv);
    \node[class] at (2,5) (tree)  {\tree};        \draw (tw<w)--(tree)--(strch);
    \draw (chbip) to[bend right] (tree);
    \node[class] at (3,5) (scs)   {\scs\,\hard\,\smix};
                                                  \draw (strch)--(scs)--(split);
    \node[class] at (1,4) (bic)   {\bic\,\smix};  \draw (conv)--(bic);
    \node[class] at (2,4) (perm)  {\perm};
    \node[class] at (3,4) (qmon)  {\qmon\,\rmix};
    \node[class] at (4,4) (int)   {\intl};
    \draw (strch) to[bend left=20] (int);
  % \node[class] at (4,4) (msp)   {~~~~~~\,\rmix};\draw (scs)--(msp);
    \node[class] at (1,3) (mono)  {\mono};        \draw (bic)--(mono)--(perm);
  % \draw[double] (mono)--(qmon);
    \node[class] at (3,3) (chp)   {\chp\,\smix};  \draw (perm)--(chp)--(int);
  % \node[class] at (4,3) (bnint) {~~~~~~};       \draw (int)--(bnint);
    \node[class] at (1,2) (Efcb)  {\Efcb\,\easy}; \draw (mono)--(Efcb);
    \node[class] at (2,2) (cogr)  {\cogr\,\nerg\,\easy};\draw (perm)--(cogr);
    \node[class] at (3,2) (qchs)  {\qchs};
    \draw (perm) to[bend right=25] (qchs);
    \draw (qmon) to[bend right=50] (qchs);
    \node[class] at (4,2) (unii)  {\unii};        \draw (int)--(unii);
    \draw (unii) to[bend right=20] (qmon);
    \node[class] at (1,1) (chain) {\chain\,\nstb};\draw (Efcb)--(chain);
    \node[class] at (2,1) (thre)  {\thre\,\nstb}; \draw (cogr)--(thre);
    \draw (int)   .. controls (3.3,1.6) .. (thre);
    \draw (scs)   .. controls (2.3,3.7) .. (2.3,3.0)
                  .. controls (2.3,2.6) .. (2.4,2.4)
                  .. controls (2.5,2.2) .. (2.5,1.8)
                  .. controls (2.5,1.5) .. (thre);
    \node[class] at (3,1) (qcbs) {\qcbs\,\stab};  \draw (qchs)--(qcbs)--(cogr);
    \node[class] at (4,1) (c-chn) {\cchn\,\easy}; \draw (unii)--(c-chn)--(chp);
  % \draw (chp)   .. controls (3.5,2.5) .. (3.5,2.0)
  %               .. controls (3.5,1.5) .. (c-chn);
    \node[class] at (1,0) (cb)    {\cb};          \draw (chain)--(cb)--(cogr);
    \node[class] at (3,0) (compl) {\compl};       \draw (qcbs)--(compl);
    \draw (thre)--(compl)--(c-chn);
    \draw (wkch)  .. controls (2.55,7.0) .. (2.55,6.0)
                  .. controls (2.55,5.7) .. (perm);
    \draw (oddch) .. controls (1.55,6.3) .. (1.55,6.0)
                  .. controls (1.55,4.5) .. (1.55,3.0)
                  .. controls (1.55,2.7) .. (cogr);
    \draw (oddch) .. controls (2.55,6.1) .. (2.4,5.5)
                  .. controls (2.25,4.9) .. (qmon);
    \draw[dotted]   (bip) .. controls (1.2,6.5) .. (2.5,6.5)
                          .. controls (3.8,6.5) .. (split);
    \draw[dotted] (chbip) .. controls (1.3,5.5) .. (2.0,5.5)
                          .. controls (2.8,5.5) .. (scs);
    \draw[dotted] (chain)--(thre);
  % \draw[dashed] (0.6,5.4) .. controls (0.5,5.4) ..
  %                 (1,5.4) .. controls (1.5,5.4) ..
  %                 (1.5,5) .. controls (1.5,4.5) ..
  %                 (2,4.5) .. controls (2.5,4.5) .. (4.4,4.5);
  %               % (3,4.5) .. controls (3.5,4.5) ..
  %               % (3.5,4) .. controls (3.5,3.5) ..
  %               % (4,3.5) .. controls (4.5,3.5) .. (4.5,3.5);
  \end{tikzpicture}
  \begin{description}[itemsep=-2pt]
  \item[\nerg] This class contains graphs on which the switch
    chain is not ergodic.
  \item[\ergo] The switch chain is ergodic on all graphs in this class.
  \item[\hard] Counting perfect matchings remains \numP-complete
    when restricted to graphs in this class.
  \item[\easy] For all graphs in this class the number of
    perfect matchings can be computed exactly in polynomial time.
  \item[\smix] This class contains a sequence of graphs on which
    the switch chain mixes slowly.
  \item[\rmix] The switch chain mixes rapidly on all graphs in
    this class.
  \item[\nstb] This class contains graphs that are not P-stable.
  \item[\stab] All graphs in this class are P-stable.
  \end{description}\vspace{-1ex}
  \caption{\label{fig:Hasse}
    Containment of graph classes.
    Dotted lines indicate linked classes.
    Double lines indicate the inclusion of a class in the \qua class of its
    closure under disjoint union.
  }
\end{figure}\vspace{-4ex}

A class $\cB$ of bipartite graphs and a class $\cS$ of split graphs are
\emph{linked} if,
\begin{enumerate}[topsep=-3pt,itemsep=-3pt,label=(\alph*)]
\item for every $G$ in $\cB$, both graphs $H$ obtained from $G$ by completing
  one of its partite set belongs to $\cS$, and
\item for each graph $H$ in $\cS$, the graph $G$ obtained from $H$ by removing
  all edges between vertices in the clique belongs to $\cB$.
\end{enumerate}
If the bipartite graph $G$ has partite sets of the same size then the extra
edges in $H$ cannot be used by any perfect matching. That is, $G$ and $H$ have
exactly the same perfect matchings. If the partite sets of $G$ differ in size
then $G$ has no perfect matching. However, $H$ might have a perfect matching
if its clique contains more vertices than its independent set.
In Fig.~\ref{fig:Hasse} dotted lines indicate linked classes.
Double lines indicate the inclusion of a class $\cC$ in $\qua\cC^*$
where the graphs in $\cC^*$ are disjoint unions of graphs in $\cC$.
For further information on these classes and references to the original work
see \cite{BrLeSp99} or \cite{Golumb04}.

\end{document}